\documentclass[11pt]{article}
\usepackage[utf8]{inputenc}
\usepackage[T1]{fontenc}
\usepackage{lmodern}
\usepackage[margin=1in]{geometry}
\usepackage{amsmath}
\usepackage{amsfonts}
\usepackage{amssymb}
\usepackage{amsthm}
\usepackage{algorithm}
\usepackage{algpseudocode}

\theoremstyle{plain}
\newtheorem{thm}{Theorem}
\newtheorem{lem}{Lemma}
\newtheorem{clm}{Claim}
\newtheorem{cor}{Corollary}
\theoremstyle{definition}
\newtheorem{Def}{Definition}

\newcommand{\cfin}{c^{\mathrm{f}}}
\newcommand{\bfin}{b^{\mathrm{f}}}

\algrenewcommand{\algorithmiccomment}[1]{\texttt{\textbackslash\textbackslash\ #1}}

\author{Riccardo Colini-Baldeschi\thanks{
		Sapienza University of Rome, Via ariosto 25,
		I-00185 Rome, Italy,\newline
		Email: \texttt{riccardo.colini@dis.uniroma1.it}
		}
        \and 
        \hspace{-1em}Monika Henzinger\thanks{
        University of Vienna, Universit\"atsstra\ss e 10/9,
        A-1090 Vienna, Austria,\newline
        Email: \texttt{monika.henzinger@univie.ac.at}}
        \and 
        \hspace{-1em}Stefano Leonardi\thanks{
        Sapienza University of Rome, Via ariosto 25,
        I-00185 Rome, Italy,\newline 
		Email: \texttt{stefano.leonardi@dis.uniroma1.it}
        }
        \and
        \hspace{-1em}Martin Starnberger\thanks{
        University of Vienna, Universit\"atsstra\ss e 10/9,
        A-1090 Vienna, Austria,\newline
        Email: \texttt{martin.starnberger@univie.ac.at}}
}

\title{On Multiple Round Sponsored Search Auctions with Budgets}

\begin{document}

\maketitle
\thispagestyle{empty}

\begin{abstract}
In a sponsored search auction the advertisement slots on a search result page  are generally
ordered by click-through rate. Bidders have a valuation, which is usually assumed to be linear in the click-through rate, a budget constraint, and receive at most one slot per search result page (round).
We study multi-round sponsored search auctions, where the different rounds are linked
through the budget constraints of the bidders and the valuation of a bidder for all rounds is the sum
of the valuations for the individual rounds. 
All mechanisms published so far
either study one-round sponsored search auctions~\cite{unit-demand,Ashlagi09positionauctions} or
the setting where every round has only one slot and all slots have the same click-through rate, which is identical to a multi-item auction~\cite{DLNcorrected}.

This paper contains the following three results: 
(1)~We give the first mechanism for the multi-round sponsored search problem where different slots have different click-through rates.
Our mechanism is incentive compatible in expectation,
individually rational in expectation, Pareto optimal in expectation, and also
ex-post Pareto optimal for each realized outcome.
(2)~Additionally we study the combinatorial setting, where each bidder
is only interested in a subset of the rounds.
We give a deterministic, incentive compatible, individually rational, and Pareto optimal mechanism for 
the setting where all slots have the same click-through rate.
(3)~We present an impossibility result for auctions where bidders have diminishing marginal valuations.
Specifically, we show that even for the multi-unit (one slot per round) setting there is no incentive compatible, individually rational, and Pareto optimal mechanism for private diminishing marginal valuations and public budgets.
\end{abstract}

\newpage
\setcounter{page}{1}

\section{Introduction}
In {\em sponsored search} (or {\em adword}) auctions advertisers bid on keywords. The  valuation of these
bidders for the  keywords is used in the generalized second price auction, which is utilized by firms such as Google, Yahoo, and Microsoft~\cite{Edelman05internetadvertising}.  The  valuation of the bidders is private  knowledge.  
Slots for ads are usually ordered  by their public click-through rate (CTR), which is decreasing by the position of the slots on a search result page (round). Each bidder is assigned  at most one slot in a round
and his valuation for a slot is assumed to depend linearly on the CTR. 
Moreover, valuations are assumed to be additive, i.e., the total valuation of a bidder is equal to the sum of his valuations for all the slots that are assigned to him.

A further key ingredient of an adword  auction is that  bidders specify a budget on the payment charged for the ads, linking the different rounds. The deterministic Vickrey auction~\cite{vickrey-61} was designed to maximize social welfare in settings where all items are different and bidders have arbitrary valuation functions. However, 
deterministic and incentive compatible auctions that maximize social welfare do not exist for the multi-unit\footnote{In the {\em multi-unit} case multiple {\em homogeneous} items are sold.} case for general valuations if there are private budget constraints.  This was shown in the seminal paper by Dobzinski, Lavi, and Nisan~\cite{DBLP:conf/focs/DobzinskiLN08,DLNcorrected}. They considered the multi-unit case with additive valuations, which in the sponsored search setting corresponds to slots with identical CTR, each round with only one slot for sale. They gave an incentive compatible auction  based on Ausubel's  ascending clinching auction \cite{Ausubel04} that produces a Pareto optimal allocation if budgets are {\em public}. 
They also showed that this assumption is strictly needed, i.e., that no deterministic mechanism for {\em private} budgets exists if we insist on incentive compatibility and on obtaining an allocation that achieves some form of efficiency.  Pareto optimality is a basic notion for the efficiency of an allocation. It seems to be the least one should aim for. If an allocation was not Pareto optimal then bidders could trade amongst themselves and improve their utilities, thus increasing efficiency.

We show that the multi-unit auction of \cite{DBLP:conf/focs/DobzinskiLN08,DLNcorrected} can be adapted to deal with {\em heterogeneous} items as in the real-life setting of adword auctions, thus, allowing us to study adword auctions with multiple slots, multiple rounds, and budgets. To the best of our knowledge we are the first to do so.
We specifically model the case of several slots with different CTR available for each round and a bound on the number of slots (usually one) that can be allocated to a bidder per round.
Since the impossibility result of~\cite{DBLP:conf/focs/DobzinskiLN08,DLNcorrected} for multi-unit auctions applies also to the adword setting, 
achieving deterministic, incentive compatible, individually rational, and Pareto optimal auctions is only possible if budgets are public. 
In this paper, we first provide a deterministic, incentive compatible, individually rational, and Pareto optimal auction that provides a fractional allocation for the case of one round where slots are divisible.  We then show how to  probabilistically round  this fractional allocation for the divisible case to an integer allocation for the indivisible case with multiple rounds (i.e., the adword setting) and get an auction that is incentive compatible in expectation, individually rational in expectation, and Pareto optimal.  

Furthermore, we address the more general case of combinatorial auctions with  bidders having a  non-zero identical  valuation only on a {\em subset} of the rounds. This case can arise, for example, if each round represents a different keyword.  Valuations are additive and each bidder is assigned at most one slot for a given round, but we restrict the  model by allowing only identical slots  for each round, i.e., we require that all slots have the same CTR. 
This setting extends the model  considered by \cite{FLSS11}
to multiple slots. We  present  a variation of the clinching auction that is deterministic,  incentive compatible, individually rational, and Pareto optimal.

Finally, we also study non-additive valuations, namely valuations with diminishing marginal valuations. 
Diminishing marginal valuation functions are widely used to model auction settings with marginal utilities being positive non-increasing functions of the number of items already 
allocated to the bidders.  We  show that even in the multi-unit (one slot per round) case there is
no  incentive compatible, individually rational, and Pareto optimal auction for private diminishing marginal valuations
and public budgets.

{\bf Technical contributions.} We briefly describe our technical contributions.

(1) The clinching auction of~\cite{DBLP:conf/focs/DobzinskiLN08,DLNcorrected} repeatedly increases the price for an item and checks whether the demand still allows to sell all 
the remaining unsold items. We extend this to adword auctions with {\em divisible} slots as follows: We define a price ``per capacity'' and give different weights to different slots depending on their CTRs.
To perform the check whether all remaining unsold items can still be sold we solve suitable linear programs. 
To show the Pareto optimality of the modified auction we give a novel characterization of Pareto optimality for adword auctions. We then show that if the outcome of the auction 
did not fulfill this new characterization then one  of the linear programs solved by the auction would not have computed an optimal solution. Since this is not possible, 
it follows that the outcome is Pareto optimal.

(2) After solving the divisible case we round the outcome for the divisible case to the indivisible case. This requires a novel swapping algorithm that guarantees that  each bidder receives at most one slot of each round.

(3) The single-valued combinatorial auction of \cite{FLSS11} solves the combinatorial auction problem over rounds, where each round has only {\em one} slot. We extend their techniques to the multi-slot per round setting as follows:
(a) We extend their B-matchings based approach by giving capacities, equal to the number of unsold slots, to nodes that represent keywords.
(b) We extend the concept of trading paths, which in turn allows us to give a characterization of Pareto optimality for the multi-slot case. 
While in the single-slot case it is sufficient to restrict the attention to simple trading paths, in our case there might be trading options where the same bidder or item can appear many times along the same path. The crucial insight is that there always exists a simple trading path whenever there exists a non-simple one.  

(4) For the impossibility proof for diminishing marginal valuations we consider a simple scenario of two bidders and two one-slot rounds, i.e., this is a two-unit setting. 
We show that under certain conditions on the marginal valuations and the budgets, any incentive compatible, individually rational, and Pareto optimal auction must assign both rounds to one of the two bidders.  However, we also present an example where the other bidder will instead receive one  of the rounds if he overbids  on the second marginal valuation.
 
{\bf Further related work.}
Ascending clinching auctions are used in the FCC spectrum auctions, see~\cite{Milgrom98putingauction,ausubel-milgrom,Ausubel04}.
For a motivation of our adword auctions see~\cite{Nisan09} on Google's auction for TV ads.
Fiat et al.~\cite{FLSS11} studied an extension of the multi-unit case of~\cite{DBLP:conf/focs/DobzinskiLN08} to a combinatorial setting where items are distinct and different bidders may be interested in different items.  The auction presented in \cite{FLSS11} is incentive compatible and Pareto optimal for additive valuations and single-valued bidders (i.e., every bidder does not distinguish between the keywords in his interest set). This result is possible only if the sets of interest are public. 
Bhattacharya et al.~\cite{DBLP:journals/corr/abs-0904-3501} showed that for one infinitely divisible item, a bidder cannot improve his utility by underreporting his budget. This leads to a randomized incentive compatible in expectation algorithm for one infinitely divisible item with both private valuations and budgets.
Aggarwal et al.~\cite{unit-demand} and Ashlagi et al.~\cite{Ashlagi09positionauctions} studied {\em envy-free} outcomes that are bidder optimal, respectively Pareto optimal in an one-round adword auction. In this setting they give (under certain conditions in~\cite{unit-demand}) an incentive compatible auction with both private valuations and budgets. Recently, Lavi and May~\cite{LM11} and D\"utting et al.~\cite{DHS11} showed impossibility results for
public budgets in various settings with heterogeneous items.

\pagebreak
In Section \ref{sec:def} we give notation and preliminary definitions.  In Section \ref{sec:fractional} we present the clinching auction for divisible non-identical slots.  
In Section \ref{sec:swapping} we show how to round a solution for the divisible case to the multi-round indivisible case. In Section \ref{sec:combinatorial} we study the single-valued combinatorial case with multiple identical slots,  and in Section \ref{sec:impossibility} we show the impossibility result for the case of diminishing marginal valuation functions. All omitted proofs are in the appendix.

\section{Problem statement and definitions}
\label{sec:def}

We have $n$ bidders and $m$ slots.
We call the set of  bidders $I:=\{1,\dots,n\}$ and the set of  slots $J:=\{1,\dots,m\}$.  Each bidder $i\in I$ has a private valuation $v_i$, a public budget $b_i$, and a public demand constraint $\kappa_i$, which is a positive integer. Each slot $j\in J$ has a quality $\alpha_j$. 
The bidders and the slots are ordered such that $v_i\geq v_{i'}$ if $i>i'$ and $\alpha_j\geq \alpha_{j'}$ if $j>j'$, where ties are  broken in some arbitrary but fixed order. 

\noindent 
{\bf The divisible case.} The goal  is to assign each bidder $i$ a fraction $x_{i,j} \geq 0$ of each slot $j$ and charge him a payment $p_i$. A matrix $X = (x_{i,j})_{(i,j)\in I\times J}$ and a payment vector $p$ are called an {\em allocation} $(X,p)$. 
We call $c_i = \sum_{j\in J} \alpha_j x_{i,j}$ the {\em weighted capacity} allocated to bidder $i$. 
An allocation  is {\em legal} if it fulfills the following conditions:
(1) $\sum_{j\in J}x_{i,j}\leq\kappa_i\ \forall i\in I$,
(2) $\sum_{i\in I}x_{i,j}=1\ \forall j\in J$, and
(3) $b_i \geq p_i \ \forall i \in I$.

\noindent 
{\bf The indivisible case.} We additionally have  a set $R$ of {\em rounds} or {\em keywords}, where $|R|$ is public. The goal is to assign each slot $j\in J$ of round $r\in R$ to one bidder $i\in I$ while obeying the demand constraint of each bidder in each round.  A three-dimensional matrix $X = (x_{i,j,r})_{(i,j,r)\in I\times J\times R}$ with  $x_{i,j,r}=1$ if slot $j$ is assigned to bidder $i$ in round $r$, and $x_{i,j,r}=0$ otherwise, and a payment vector $p$ form an allocation $(X,p)$. We call $c_i = \sum_{j\in J} \frac{\alpha_{j}}{|R|} (\sum_{r\in R}x_{i,j,r})$ the {\em weighted capacity} allocated to bidder~$i$.  An allocation is {\em legal} if it fulfills the following conditions:
(1) $\sum_{j\in J}x_{i,j,r}\leq \kappa_i\ \forall i\in I, \forall r\in R$,
(2)  $\sum_{i\in I}x_{i,j,r}=1\ \forall j\in J, \forall r \in R$, and
(3) $b_i \geq p_i \ \forall i \in I$.

\noindent 
{\bf The combinatorial indivisible case.}  In the combinatorial case not all rounds respectively keywords are identical. 
Every bidder $i\in I$ has a publicly known set of interest $S_i \subseteq R$.  
This corresponds to valuation $v_i$ for all keywords in $S_i$ and valuation $0$ for all other keywords. 
We model this case by imposing  $x_{i,j,r}=0\ \forall r\notin S_i$.

\noindent 
{\bf Properties of the auctions.}   The {\em utility} of bidder $i$ for the legal allocation $(X,p)$ is defined by $u_i = c_i v_i - p_i$.  The allocation must obey the following conditions: 
({\em Bidder rationality}) $u_i\geq 0$ for all bidders $i\in I$,
({\em Auctioneer rationality}) the utility of the auctioneer fulfills $\sum_{i\in I} p_i \geq 0$, and
({\em No positive transfer}) $p_i\geq 0$ for all bidders $i\in I$.
An allocation that is both bidder rational and auctioneer rational is called {\em individually rational.}
An allocation $(X,p)$ is {\em Pareto optimal} if there is no other allocation $(X',p')$ such that 
(1) the utility of none bidder in $(X,p)$ is less than his utility in $(X',p')$, 
(2) the utility of the auctioneer in $(X,p)$ is no less than his utility in $(X',p')$, and
(3) at least one bidder or the auctioneer is better off in $(X',p')$ compared with $(X,p)$.
An auction is {\em incentive compatible} if it is a dominant strategy for all bidders to reveal their true valuation. An auction is said to
be {\em Pareto optimal} if the allocation it produces is Pareto optimal. 
A randomized auction is {\em Pareto optimal in expectation} if the above conditions hold in expectation.

We show that our randomized mechanism for indivisible slots is Pareto optimal in expectation and that each realized outcome is Pareto optimal. Note that neither of these conditions implies each other. Let us assume that we have two bidders, a single indivisible item, and a uniformly distributed random variable $Y\sim\mathcal{U}(0,1)$. Consider first the case that bidder~1 has valuation $v_1=1$ and budget $b_1=1$, bidder~2 has valuation $v_2=2$ and budget $b_2=1$, and we have a value $\tilde{y}\in (0,1)$. If we sell the item to the bidder~2 for price $p_2=1$ (and $p_1=0$) for every realization $y$ of $Y$ with $y\neq \tilde{y}$ the outcome is Pareto optimal in expectation. However, only if we sell the item to bidder~2 also for $y=\tilde{y}$ every possibly realized outcome is Pareto optimal. Hence, Pareto optimality in expectation does not imply that each realized outcome is Pareto optimal. Next consider the case that bidder~1 has valuation $v_1=1$ and budget $b_1=1$, and bidder~2 has valuation $v_2=2$ and budget $b_2=0.5$. If we sell the item to bidder~1 for price $p_1=1$ (and $p_2=0$) for every realization $y\in (0,1)$ each realized outcome is Pareto optimal because $v_1>b_2$. However, we could select the bidder who gets the item with probability one half, and both bidders have to pay $p_1=p_2=0.5$ independent of the assignment. Hence, the outcome is not Pareto optimal in expectation, and therefore, Pareto optimality in expectation is not implied if every realized outcome is Pareto optimal. 

\section{Deterministic clinching auction for the divisible case}
\label{sec:fractional}

\subsection{Characterization of Pareto optimality}
\label{sec:paretochar}

Given a legal allocation $(X,p)$, a {\em swap} between two bidders $i$ and $i'$ is a fractional exchange of slots, i.e., if there are slots $j$ and $j'$ and a constant $\tau > 0$ with $x_{i,j} \geq \tau$ and $x_{i',j'} \geq \tau$ then
a swap between $i$ and $i'$ gives a new legal $(X',p)$ with $x'_{i,j} = x_{i,j} - \tau$, 
$x'_{i',j'} = x_{i',j'} - \tau$, $x'_{i,j'} = x_{i,j'} + \tau$, and $x'_{i',j} = x_{i',j} + \tau$. If $\alpha_j < \alpha_{j'}$ then the swap increases $i$'s weighted capacity.
We assume throughout this section that the number of slots $m$ fulfills $m=\sum_{i\in I} \kappa_i$.
To characterize Pareto optimal allocations we first define for each bidder $i$ the set $N_i$ of bidders such that
for every bidder $i'$ in $N_i$ ~there exists a swap between $i$ and $i'$ that increases $i$'s weighted capacity.
Given a legal allocation $(X,p)$ we use $h'(i) := \max \{ j\in J| x_{i,j}> 0\}$ for the slot with the highest quality that is assigned to bidder $i$ and $l(i) := \min \{j\in J| x_{i,j}> 0\}$ for the slot with the lowest quality that is assigned to bidder $i$. To consider the case of slots with equal $\alpha$-value we define $h(i):=\min\{ j\in J| \alpha_j=\alpha_{h'(i)}\}$.
Now, $N_i=\{ a \in I|  h(a) > l(i)\}.$
To model sequences of swaps we define furthermore
$N^k_i=N_i$ for $k = 1$ and $N^k_i= \bigcup_{a \in N^{k-1}_i}N_a$ for $k > 1$.
Since we have only $n$ bidders, $\bigcup_{k=1}^n N^k_i=\bigcup_{k=1}^{n'} N^k_i$ for all $n'\geq n$. We define $\tilde{N}_i:=\bigcup_{k=1}^n N^k_i \setminus \{i\}$ and $\tilde{v}_i=\min_{a\in\tilde{N}_i} (v_a)$ if $\tilde{N}_i\neq \emptyset$ and $\tilde{v_i}=\infty$ else. 
Given a legal allocation $(X,p)$ we use  $B:=\{ i\in I| b_i>p_i\}$ to denote the set of  bidders who have a positive remaining budget.

\begin{thm}\label{lem:char1}
If $\tilde{v_i}\geq v_i\ \forall i\in B$ then the respective legal allocation $(X,p)$ is Pareto optimal.
\end{thm}

We say that a legal allocation $(X,p)$ contains a {\em trading swap sequence} (for short {\em trading swap}) if there exists a legal allocation $(X', p')$ and two bidders $u, w\in I$ such that
(1) bidder $w \in \tilde N_u$,
(2) for all $i \in I \setminus \{u,w\}$ 
it holds that $\sum_{j \in J} \alpha_j x_{i,j} = \sum_{j \in J} \alpha_j  x'_{i,j}$ and $p_i = p'_i$,
(3) $\delta := \sum_{j \in J} \alpha_j (x'_{u,j} - x_{u,j} )= \sum_{j \in J} \alpha_j (x_{w,j} - x'_{w,j}) > 0$,
(4) $v_u \delta > p'_u - p_u = p_w - p'_w = v_w \delta$, and
(5) $b_u \geq p'_u$.
We say that the allocation $(X',p')$ {\em results from the trading swap.}

\begin{thm}\label{lem:char2}
(a) Given a legal allocation $(X,p)$ such that $\exists u \in B:\ \tilde{v}_u<v_u$ then there exists a trading swap in $(X,p)$. (b) If there exists a trading swap in $(X,p)$ then the allocation $(X',p')$ resulting from the trading swap is a legal allocation that is Pareto superior to $(X,p)$.

\end{thm}
\begin{proof}
We know that there is a bidder $u \in B$ with $\tilde{v}_u<v_u$. Thus, we can select the smallest $k\in \{1,\dots, n\}$ for which there is a bidder $a_k\in N^k_u$ who has $v_{a_k} = \tilde{v}_u$. We define for all $p\in \{1,\dots,k-1\}$ the bidder $a_p$ such that $a_p\in N^p_u$ and $a_{p+1}\in N_{a_p}$ and set $a_0:= u$.
Since we selected the smallest $k$, we know that $a_p\neq a_{p'}$ if $p\neq p'$. The fact that $a_{p+1}\in N_{a_p}$ implies that $h({a_{p+1}})>l(a_p)$.
Hence, we could swap a fraction of size $\epsilon_{p+1}:=\min\{x_{a_{p},l(a_{p})},x_{a_{p+1},h({a_{p+1}})}\}$ of the slots $h({a_{p+1}})$ and $l({a_p})$ between the bidders $a_{p+1}$ and $a_{p}$ with $p\in\{ 0,\dots,k-1\}$.
Such a swap increases the weighted capacity that is assigned to bidder $a_p$ by $\delta_{p+1}:=\epsilon_{p+1}(\alpha_{h({a_{p+1}})}-\alpha_{l({a_{p}})})$, while the weighted capacity that is assigned to bidder $a_{p+1}$ is decreased by $\delta_{p+1}$.
We define $\delta:=\min (\{ \frac{b_{a_0}-p_{a_0}}{v_{a_k}}\}\cup\{\delta_p|p\in\{ 1,\dots,k\}\})$ and $\tau_{p+1}:=\frac{\delta}{\alpha_{h({a_{p+1}})}-\alpha_{l({a_{p}})}}\ \forall p\in\{ 0,\dots,k-1\}$ and define an allocation $(X',p')$ as follows: We set $x'_{a_p,h({a_{p+1}})}:=x_{a_p,h({a_{p+1}})}+\tau_{p+1}$ and $x'_{a_p,l({a_{p}})}:=x_{a_p,l({a_{p}})}-\tau_{p+1}$ for all $p\in\{ 0,\dots,k-1\}$, $x'_{a_{p},h({a_{p}})}:=x_{a_{p},h({a_{p}})}-\tau_{p}$ and $x'_{a_{p},l({a_{p-1}})}:=x_{a_{p},l({a_{p-1}})}+\tau_{p}$ for all $p\in\{ 1,\dots,k\}$, and $x'_{i,j} = x_{i,j}$ for all other $(i,j)\in I\times J$.
Moreover, we set $p'_{a_k}:=p_{a_k}- v_{a_k}\delta$, $p'_{a_0}:=p_{a_0}+ v_{a_k}\delta$, and $p'_i:=p_i$ for all other $i\in I$. Thus, with $w = a_k$ it follows that $(X',p')$ fulfills conditions (1)-(5) of a trading swap.

Next we show that $(X',p')$ is a legal allocation.
By the definition of $X'$ for all $i \in I$ it holds that $\sum_{j\in J} x'_{i,j} = \sum_{j\in J} x_{i,j} = \kappa_i$ as
whenever for some $\tau$ with $-1\leq \tau\leq 1$, $x_{i,j}'$ is set to $x_{i,j} + \tau$ for some $j \in J$, $x_{i,l}'$ is set to $x_{i,l} - \tau$ for some other $l \in J$.
Additionally for every $j \in J$ it holds that $\sum_{i\in I} x_{i,j}' = \sum_{i\in I} x_{i,j} = 1$ as whenever 
$x_{a_p,j}'$ is set to $x_{a_p,j} + \tau$ for some $\tau$ with $-1\leq\tau\leq 1$, either $x_{a_{p+1}, j}'$ is set to $x_{a_{p+1}, j} - \tau$, or $x_{a_{p-1}, j}'$ is set to $x_{a_{p-1}, j} - \tau$.
Finally, $p'_i \leq p_i \leq b_i$ for all $i \not= u$ and by our construction $p'_u\leq b_u$.
 This shows that conditions (1) - (3) of a legal allocation hold for $(X',p')$.

The proof of (b) follows directly from the definition of a trading swap since the utility of no bidder or the auctioneer is decreased, but the utility of bidder~$u$ is increased.
\end{proof}

\subsection{Multiple round auction for the divisible case}
\label{sec:auction} 

In this section, we describe our deterministic auction for divisible slots and show that it is Pareto optimal, individually rational, and incentive compatible. 
A formal description is given in the procedure \textsc{Auction}  and the procedure \textsc{Sell}.
We assume throughout this section that (1) each bidder $i\in I$ has a private valuation $v_i\in \mathbb{N}_{+}$\footnote{All the arguments go through if we simply assume that $v_i \in \mathbb{Q_{+}}\ \forall i\in I$ and
there exists a publicly known value $z\in \mathbb{R}_{+}$ such that $v_i \geq z$ for every bidder $i\in I$, and for all bidders $i$ and $i'$ either $v_i = v_{i'}$ or $|v_i - v_{i'}| \geq z$.}, a public budget constraint $b_i\in \mathbb{Q}_{\geq 1}$, and a public demand constraint $\kappa_i\in \mathbb{N}_{+}$, and (2) each slot $j\in J$ has a public quality $\alpha_j\in\mathbb{Q}_{+}$. From now on the order of the bidders is independent of their valuations and their budgets.

\begin{algorithm}
\caption{Clinching Auction for Divisible Slots.}
\label{alg:auction}
\begin{algorithmic}[1]
\Procedure{Auction}{$I, J, \alpha,\kappa, v, b$}

	\State $A \gets I;\ p_i \gets 0\ \forall i\in I$
	\State \Comment{$c_i$ is the minimal capacity achievable by agent~$i$}
	\State $c_i \gets \min_{J'\subseteq J, |J'|=\kappa_i} (\sum_{j\in J'} \alpha_j)\ \forall i\in I$
	\State \Comment{$e+c_i$ is an upper bound on the capacity achievable by agent~$i$}
	\State $e \gets \max_{J'\subseteq J, J''\subseteq J, |J'|=|J''|}( \sum_{j\in J'} \alpha_j -\sum_{j\in J''} \alpha_j )$
	\State $\pi\gets \frac{1}{\max\{ 1, e \}};\ \pi^{+}\gets \lfloor \pi\rfloor+1;\ d_i\gets\frac{b_i}{\pi}\ \forall i\in I$
	\While{$\sum_{i\in I} c_i < \sum_{j\in J} \alpha_j$}
	\State $E \gets \{i\in A| \pi^{+} > v_i\}$
	\For{$i\in E$}
	\State $(X,s)\gets \text{\textsc{Sell}}(I,J,\alpha,\kappa,v,c,d,i)$
	\State $(c_i,p_i,d_i)\gets(c_i+s,p_i+s\pi,0)$
	\EndFor
	\State $A\gets A \setminus E$
	\State{$d^{+}_i \gets \frac{b_i-p_i}{\pi^{+}}\ \forall i\in A$}
	\While{$\exists i\in A\ \text{with}\ d_i\neq d^{+}_i$}
	\State $i' \gets \min(\{ i\in A | d_i \neq d_i^{+}\} )$
	\For{$i\in A\setminus i'$}
	\State $(X,s)\gets\text{\textsc{Sell}}(I,J,\alpha,\kappa,v,c,d,i)$
	\State $p_i \gets \begin{cases}
	p_i+s\pi,\ &\text{if}\ d_i \neq d^{+}_i\\
	p_i+s\pi^{+},\ &\text{else}\end{cases}$
	\State $(c_i, d_i, d^{+}_i)\gets (c_i + s, d_i -s, \frac{b_i - p_i}{\pi^{+}})$ 
	\EndFor
	\State $(X,s)\gets\text{\textsc{Sell}}(I,J,\alpha,\kappa,v,c,d,i')$	
	\State $(c_{i'},p_{i'})\gets (c_{i'}+s,p_{i'}+s\pi)$
	\State $d^{+}_{i'}\gets \frac{b_{i'} - p_{i'}}{\pi^{+}};\ d_{i'} \gets d^{+}_{i'}$
	\EndWhile
	\State $\pi\gets\pi^{+};\ \pi^{+}\gets\pi^{+}+1$
	\EndWhile
	\State \textbf{return} $(X,p)$
\EndProcedure
	\end{algorithmic}
	\end{algorithm}

\begin{algorithm}
\caption{Determine the amount that bidder $i'$ clinches.}
\label{alg:sell}
\begin{algorithmic}[1]
\Procedure{Sell}{$I, J, \alpha, \kappa, v, c, d, i'$}
	\State compute an optimal solution of the following linear program\newline
	 \mbox{\hspace{1.5em}}that is a vertex of the polytope defined by its constraints:\newline
	$\begin{array}{lrlrcll}
	\quad&\multicolumn{2}{r}{\text{minimize}}&\gamma_{i'}&&&\\
	&\text{s.t.:}&\text{(a)}&\sum_{i\in I} x_{i,j} &\hspace{-10pt}=\hspace{-10pt}& 1&\forall j\in J\\
	&                      &\text{(b)}&\sum_{j\in J} x_{i,j} &\hspace{-10pt}=\hspace{-10pt}& \kappa_i&\forall i\in I\\
	&                      &\text{(c)}&\sum_{j\in J} x_{i,j}\ \alpha_j - \gamma_i &\hspace{-10pt}=\hspace{-10pt}& c_i&\forall i\in I\\
	&                      &\text{(d)}&\gamma_i           &\hspace{-10pt}\leq\hspace{-10pt}& d_i&\forall i\in I\\
	&                      &\text{(e)}&x_{i,j}           &\hspace{-10pt}\geq\hspace{-10pt}& 0&\forall i\in I, \forall j\in J\\
	&                      &\text{(f)}&\gamma_i           &\hspace{-10pt}\geq\hspace{-10pt}& 0&\forall i\in I\\
	\end{array}$
	\State \textbf{return} $(X,\gamma_{i'})$
\EndProcedure
	\end{algorithmic}
	\end{algorithm}

Before we start the auction, we add $\sum_{i\in I} \kappa_i - |J|$ dummy-slots with quality zero to the set of slots~$J$ if $\sum_{i\in I}\kappa_i>|J|$ and we remove the $|J|-\sum_{i\in I} \kappa_i$ slots with the lowest quality $\alpha_j$ out of the set of slots~$J$ if $\sum_{i\in I}\kappa_i<|J|$. 
Thus, the demand constraints $\sum_{j\in J} x_{i,j}\leq \kappa_i\ \forall i\in I$ and the condition $\sum_{i\in I} x_{i,j} = 1\ \forall j\in J$ imply that each legal allocation fulfills even $\sum_{j\in J} x_{i,j}= \kappa_i\ \forall i\in I$.
Furthermore, we have to assign to each bidder $i\in I$ an initial weighted capacity of $c_i = \min_{J'\subseteq J, |J'|=\kappa_i} (\sum_{j\in J'} \alpha_j)$ for a price of zero at the beginning of the auction. If some bidder~$i$ obtained less weighted capacity at termination, we could not assign all the slots fully to the bidders. 

The state of the auction is defined by the current price $\pi$, the weighted capacity $c_i$ that bidder $i\in I$ has clinched so far, and the payment $p_i$ that has been charged so far to bidder $i$.
Furthermore, we have the price of the next iteration $\pi^{+}$ that is $\pi^+= \pi +1$ with the exception of the  initial step when we set $\pi^+= \lfloor \pi\rfloor +1$. 
Note that the initial choice of~$\pi$ in line~7 guarantees that $\pi\leq 1$. Since we assume that $v_i\geq 1\ \forall i\in I$ and $b_i\geq 1\ \forall i\in I$, it follows that every agent can afford to buy at this price at least $e$ weighted capacity. The value $e$ is defined such that for all $i\in I$ the sum $e+c_i$ is an upper bound of the weighted capacity that can be acquired by bidder~$i$. Following the clinching auction in~\cite{DBLP:journals/corr/abs-0904-3501,DLNcorrected}
we do not increase the price that a bidder $i\in I$ has to pay from $\pi$ to $\pi^+$ for {\em all} bidders at the same time. Instead, we call \textsc{Sell} each time before we increase the price for a single bidder. We define the set of {\em active} bidders $A\subseteq I$ which are all those $i\in A$ with $\pi \leq v_i$, and the subset $E$ of $A$ of {\em exiting} bidders which are all those $i\in A$ with $\pi^+ > v_i$. If the price that bidder $i\in A$ has to pay for a unit is $\pi$ then his demand is $d_i=\frac{b_i-p_i}{\pi}$. If the price he has to pay was already increased to $\pi^+$ then his demand is $d_i=\frac{b_i-p_i}{\pi^+}$. In this case, the demand corresponds to $d_i^+$ that is always equal to $\frac{b_i-p_i}{\pi^+}$. 
Different from the auction in~\cite{DBLP:journals/corr/abs-0904-3501,DLNcorrected} a bidder with $d_i = d_i^+$
is also charged the increased price $\pi^+$ if he receives additional weighted capacity.
Since no bidder will ever pay more than his reported valuation and the demand is set so that
$b_i \geq p_i$, individual rationality follows.

Let us now detail the idea of the auction and explain why it is incentive compatible. The crucial point of the auction is that we do only sell the weighted capacity $s$ that we compute in \textsc{Sell} to the bidder~$i$ at a certain price $\pi$ or $\pi^+$ if we cannot sell $s$ to the other bidders. We check that condition by solving the linear program in \textsc{Sell}. For each iteration of the outer while-loop we first call \textsc{Sell} for each exiting bidder $i$ and sell him $s$ for price $\pi$. This is the last time when he can gain weighted capacity. Afterward, he is no longer an active bidder. Next, we call \textsc{Sell} for each of the remaining active bidders that are not exiting and sell them the respective $s$. Afterward, we increase the price of a single active bidder to $\pi^+$ and call \textsc{Sell} for every active bidder again. We continue until the price of each active bidder is increased to $\pi^+$. We can now set $\pi$ to $\pi^+$ and $\pi^+$ to $\pi^+ +1$. 

By the construction of the auction, each bidder $i$ pays never more than his reported valuation. If his reported valuation is $\tilde{v}_i$ and $\tilde{v}_i<v_i$, he becomes an inactive bidder at the price $\tilde{v}_i$. His utility cannot increase as he gets the same weighted capacity for each price $\pi<\tilde{v}_i$ and not more for price $\tilde{v}_i$. If his reported valuation is $\tilde{v}_i>v_i$, he gets the same weighted capacity for each price $\pi<v_i$ and his utility cannot increase as well. Thus, the auction is incentive compatible.

During the auction, procedure \textsc{Sell} guarantees us the existence of legal allocations. The procedure solves a linear program. As discussed above, the initial price is set so low that the first call to \textsc{Sell}
is guaranteed to have a feasible solution since each bidder $i\in I$ has enough budget to buy the weighted capacity $\gamma_i$ in every possible $(X,\gamma)$ with $\sum_{j\in J}x_{i,j}=\kappa_i$. 
Thus, there exists an optimal solution $(X,\gamma)$, where nothing is assigned to a certain bidder~$i'$ (i.e., $\gamma_{i'}=0$). 

If $(X,\gamma)$ is an optimal solution of the linear program in a call to \textsc{Sell} that minimizes $\gamma_{i'}$, then $c_{i'}$ is increased by $\gamma_{i'}$ after \textsc{Sell}, and thus, $(X,\tilde{\gamma})$ with $\tilde{\gamma}_i=\gamma_i$ for $i\neq i'$ and $\tilde{\gamma}_{i'}=0$ for $i=i'$ is a feasible solution of the linear program in the next call to \textsc{Sell}, which uses the new $c$-values. Since $\tilde{\gamma}_{i'}=0$, $(X,\tilde{\gamma})$ is a feasible solution in the next call even if the price for bidder~$i'$ was increased, and thus, his demand was decreased. A repeated application of this argument shows that
the final assignment $X$ and $\gamma=0$ is a feasible solution of the linear program in \textsc{Sell} at the conclusion of the auction. As $b_i \geq p_i$ for
every bidder $i$, it follows that the allocation $(X,p)$ computed by the auction is legal.

All the coefficients of the affine functions used in the constraints of the first linear program that gets solved during the auction are rational numbers and all the linear programs have feasible solutions. Thus, there exists an optimal solution that is a vertex of the polytope that is defined by the constraints of the respective linear program. Since that optimal solution lies on the intersection of the graphs of affine functions with rational coefficients it follows that the selected optimal solution $(X,\gamma)$ has only rational entries. The prices are rational numbers as well, and thus, $c_i$ and $d_i$ are rational numbers for all $i\in I$ in the next iteration. Hence, the allocation $(X,p)$ that is determined by the auction has only rational entries. We use that property in the next section.

\label{sec:paretoproof}

We show finally that the allocation $(X,p)$ our auction computes does not contain any trading swap, and thus, by Theorem~\ref{lem:char1} and~\ref{lem:char2} it is Pareto optimal. The proof shows that every trading swap in $(X,p)$ would lead to a superior solution to one of the linear programs solved by the mechanism. Since the mechanism found the optimal solution this leads to a contradiction.

\begin{thm}\label{lem:tp} 
The allocation $(X,p)$ produced by the auction does not contain any trading swap.
\end{thm}

\section{Multiple round randomized auction for the indivisible case}
\label{sec:swap}
\label{sec:swapping}

We will now use the allocation computed by the deterministic auction for {\em divisible} slots 
to give a randomized auction for $|R|$ rounds and {\em indivisible} slots. The randomized auction has to assign to every slot $j\in J$ exactly one bidder $i\in I$ for each round $r\in R$. 

Given an input for the indivisible case we can use it as an input for the divisible case;
the set of bidders $I$ and their $\kappa_i$ and $b_i$, and the set of slots $J$ and their $\alpha_j$ stay unchanged. Based on the allocation $(X,p)$ for the divisible problem we construct a matrix $M'$ of size $|J| \times \lambda$, where $\lambda$ is the least common denominator of all the $x_{i,j}$ values and where each column of $M'$ corresponds to a legal assignment for the indivisible one-round case. We then pick $|R|$  times one column  uniformly at random from the columns of $M'$. The $r$-th such column gives the assignment of bidders to slots for round $r$.
All the columns together form the $|J| \times |R|$ matrix $N$. We show that the randomized 
auction is (a) individually rational in expectation, (b)   incentive compatible in expectation,
(c) Pareto optimal in expectation,  and (d)  each realized outcome is ex-post Pareto optimal.

We use the payment determined for the divisible problem as payment for the allocations in the indivisible case and show below that the expected weighted capacity of each bidder equals
his weighted capacity in the divisible problem. Thus, the expected utility of each bidder in the indivisible case is equal to his utility in the derived divisible case ({\em Utility equivalence}). It follows that the randomized auction is individually rational in expectation and incentive compatible in expectation.

\begin{lem}
\label{prp:randpo}
For every probability distribution over legal allocations in the indivisible case there exists a legal allocation
$(X,p)$ in the divisible case such that the utility of the bidders and the auctioneer equals their expected
utility using this probability distribution.
\end{lem}

By Lemma~\ref{prp:randpo} any probability distribution over
legal allocations in the indivisible case that is Pareto superior to the distribution generated by our auction would
lead to a legal Pareto superior allocation for the divisible case. 
This is not possible since the allocation computed by our auction for the
divisible case is Pareto optimal. This shows that the allocation for the indivisible case is Pareto optimal in expectation.  Moreover, every realized outcome for the indivisible case is ex-post Pareto optimal: if in the indivisible case there would exist a Pareto superior allocation to one of the allocations  that gets chosen with a positive probability, then a Pareto superior expected allocation would exist in the indivisible case.

We describe next how to convert the legal allocation of the derived divisible problem into a legal allocation of the indivisible problem such that utility equivalence is guaranteed. We first discretize the allocation $X$ as follows. Recall that all $x_{i,j}$ are rational numbers. Let $\lambda$ be  their least common denominator,
set $C=\{ 1,\dots, \lambda\}$ and set $y_{i,j}=\lambda x_{i,j}$. Since $\sum_{i\in I} x_{i,j} = 1$
and $\sum_{j\in J} x_{i,j} \leq \kappa_i$, we know that $\sum_{i\in I} y_{i,j} = \lambda$ and $\sum_{j\in J} y_{i,j}
\leq \lambda \kappa_i$. 
We construct a matrix $M$ of size $|J|\times |C|$ with values in $I$ by setting $y_{i,j}$ values of row $j$ to $i$. More formally, for each $j \in J$ and each $c \in C$ we set entry $m_{j,c}=v$ for the unique value $v$ with $\sum_{i=1}^{v-1} y_{i,j}<c$ and $\sum_{i=1}^{v} y_{i,j}\geq c$. As a result $|\{c\in C| m_{j,c} = i\}| = y_{i,j}\ \forall i\in I,\forall j\in J$. 
The demand constraints imply that there are at most $|C|\kappa_i$ entries in $M$ that have the value~$i\in I$. 
Next, we replace each bidder $i$ by $\kappa_i$ {\em pseudo-bidders} and translate $M$ into a matrix $\tilde M$ such that no pseudo-bidder has more than $|C|$ entries in $\tilde M$. Then we apply to $\tilde M$ a swapping algorithm that guarantees that (1) in each column in $\tilde M$ there is at most one entry for each pseudo-bidder, and (2)
for each $j\in J$ each value appears as often in row $j$ of $\tilde M$ as it does in $\tilde M'$. Thus, when we convert all the entries of the pseudo-bidders of a given bidder $i$ into entries for bidder $i$ we get a matrix $M'$ such that each bidder $i$ has at most $\kappa_i$ entries in each column and for all $i \in I$ and $j \in J$
it holds that $|\{c\in C| m'_{j,c} = i\}| = y_{i,j}$. Details are given in the appendix.
The columns of  matrix $M'$ are then used in the sampling step.

\begin{thm}\label{thm:swap}
Given a matrix $M$ of size $|J|\times |C|$ with entries valued in $I$  and where each value appears in at most $|C|$ entries, there exists a swapping algorithm that finds a matrix $M'$ with the same size and where (1) each value appears as often in row $j$ of $M'$ as it appears in row $j$ of $M$ and (2) each value appears in at most one entry of each column of $M'$.
\end{thm}

The allocation of the randomized auction for multiple rounds and indivisible slots is now constructed as follows. We assign randomly and with equal probability $\frac{1}{|C|}$ one of the columns of matrix $M'$ to the $r$-th column of matrix $N$ for each column respectively round $r\in R$. 
We set the random variable $z_{i,j,r}$ to 1 if bidder $i$ is assigned slot $j$ in round $r$, and to 0 otherwise.
The expected weighted capacity allocated to bidder $i\in I$ is thus
\begin{equation*}
\mathbf{E}(\sum_{j\in J} \frac{\alpha_{j}}{|R|} \sum_{r\in R} z_{i,j,r})  
= \sum_{j\in J} \frac{\alpha_{j}}{|R|} \sum_{r\in R}\mathbf{E} z_{i,j,r} 
=  \sum_{j\in J} \alpha_{j} \frac{y_{i,j}}{|C|} 
=  \sum_{j\in J} \alpha_j x_{i,j}.
\end{equation*}
This proves utility equivalence. Additionally, all of the slots are fully assigned to the bidders, and hence, the stated properties are fulfilled by the randomized auction.

\section{The single-valued combinatorial case with multiple slots}

\label{sec:combinatorial}

In this section we consider  single-valued combinatorial auctions with multiple identical slots in multiple rounds.  We interpret the different slots in a round as multiple instances of the same item. Every bidder $i\in I$ has valuation 
$v_i$ on all rounds of his preference set $S_i$.   All other rounds are valued zero.   The preference sets $S_i$  and the budgets $b_i$ are public knowledge.  
We further restrict to the case of at most one slot per round allocated to a single bidder, i.e., $\kappa_i=1$.  We also require that at least $m$ bidders are interested in each round. 

A feasible allocation  $(H,p)$ is characterized by a tuple  $H=(H_{1},H_{2},\ldots,H_{n})$ where $H_i\subseteq S_i$ represents the set of items that are allocated to bidder $i$,   and by a vector of payments $p=(p_{1},p_{2},\ldots,p_{n})$  with $p_{i}\leq b_i$ for all $i$  in $I$.   The utility of bidder $i$ is defined by $u_i:=v_i |H_i| - p_{i}$. The utility of the auctioneer is  $\sum_{i=1}^{n} p_{i}$. In the combinatorial case we base the allocation of the items in the clinching auction on $B$-matchings computed on a bipartite graph $G$ with the vertex set 
$I \cup R$ and the edge set  $\{(i, t)\in I\times R| t\in S_i\}.$ The $B$-matchings are subgraphs of $G$ with maximal weight that fulfill the degree constraints of the vertices that are defined in the constraint vector $B$.
We slightly abuse notation by using $H$ to denote both: the sets of the items that are allocated to the bidders, and the $B$-matching that describes the allocation  in graph  $G$.  

Pareto optimality has been related in previous work \cite{FLSS11,DBLP:conf/focs/DobzinskiLN08}  to the non-existence of trading options between bidders.  We need a new definition of a trading path because we consider multisets of items. 

\begin{Def}
A path  $\sigma = (a_1, t_1, a_2, t_2, \ldots, a_{j-1}, t_{j-1}, a_j)$ is an alternating path with respect to an assignment $H$ if $(a_i,t_i) \in H$, $t_i \in S_{i+1}$, and $t_{i} \not\in H_{i+1}$ for all $1 \leq i<j$. 
\end{Def}

\begin{Def}
A path $\sigma = (a_1, t_1, a_2, t_2, \ldots, a_{j-1}, t_{j-1}, a_j)$ is a \emph{trading path} with respect to allocation
$(H,p)$ if the following holds: (1)  $\sigma$ is an  alternating path in $H$,  (2)  the valuation  of bidder $a_j$ is strictly greater than the
valuation of bidder $a_{1}$ (i.e., $v_{a_j}>v_{a_{1}}$), (3) the remaining (unused) budget $b^*_{a_j}$ of bidder $a_j$ at the conclusion of the
auction is at least the valuation of bidder $a_{1}$ (i.e.,~$b^*_{a_j}\geq v_{a_{1}}$).
\end{Def}

Observe that the condition $t_{i} \not\in H_{i+1}$ is needed in this case since the slots of a round have to be assigned to different bidders.  This is not the case
in the definition of alternating paths given in~\cite{FLSS11}.  Moreover, we do not restrict in principle to simple alternating paths (without cycles) as in~\cite{FLSS11}. 
We call two alternating paths Pareto equivalent if they have the same start and end bidders and produce the same change in weighted capacity for all the bidders. We are able to prove the following lemma.
\begin{lem}\label{lem:simple-complex}
If there exists an alternating path $\sigma = (a_1, t_1, \ldots, t_{j-1}, a_j)$ that contains cycles there exists also a Pareto equivalent simple alternating path.
\end{lem}

Pareto optimality and simple trading paths are now related by the following theorem.
\begin{thm}
\label{thm:paropt}
Any allocation $(H,\, p)$ is Pareto optimal if and only if (1) all slots of the rounds are sold in $(H,\, p)$, and (2)  
there are no simple trading paths in $(H,\, p)$.
\end{thm}

We define the auction in Algorithm~\ref{alg:combclinch}. During the execution of the algorithm  there is always a price $\pi$ (initially
zero), a set of unsold items $R$ (i.e., of items with unsold instances) of cardinality ${\bar r}=\left|R\right|$, 
a vector of remaining budgets $b=(b_{1},b_{2},\ldots,b_{n})$, and a vector of the number of unsold slots that are instances of the same item $(c_{1}, c_{2}, \ldots, c_{\bar{r}})$.
We denote by $U$ the multiset formed by the multiset-union of the unsold slots of all items and by $\bar t=\sum_{i=1}^{\bar r} c_i$ its cardinality.
The current demand of  bidder $i$ during the course of the auction is the number of slots that bidder~$i$ could clinch at price~$\pi$ and
is denoted by $d_i$. It is either equal to $D_i$ or to $D_i^+$, which are defined as follows:
\begin{align*}
D_i (\pi,{\bar r},b)&:=\begin{cases}
\min \{\bar{r},|S_i|,\lfloor \frac{b_i}{\pi}\rfloor\},&\text{if}\ \pi\leq v_i\\
0,&\text{else}
\end{cases}
&  
D_i^+(\pi,{\bar r},b) &:=  \lim_{\epsilon \rightarrow 0^+} D_i(\pi+\epsilon,{\bar r},b)
\end{align*}

We now  define the set $A:=\{ i\in I | D_i>0\}$ \label{eq:A} of bidders with positive demand  and the subset $E:=\{ i\in I| D_{i}>0 \wedge v_{i}=\pi\}$ \label{eq:V}  of the bidders in $A$   with valuation equal to the current price.
Bidders in  $A$ are called {\em active bidders} whereas bidders in $E$ are called {\em exiting bidders}.

The auction continues, as long as there is a bidder that belongs to $A$. At every price $\pi$ we first try to sell slots 
to any exiting bidder because even if the utility of the exiting
bidder does not increase with the new item, the utility of the auctioneer
will. After this check, we have to verify if any bidder can clinch
any slot and eventually sell that slot to him.  We denote by  $B(\{\neg i\})$ the number of slots assigned to bidders other than $i$ in a  maximal $B$-matching, and assign to bidder~$i$ the minimal number of items amongst all the maximal $B$-matchings.  An item is clinched by bidder~$i$ when $B(\neg\{i\})<\bar t$. 
If no item is clinched, we set $d_{i}=D_{i}^{+}$ for a bidder $i$ with $d_{i}>D_{i}^{+}$, and this loop continues
until no bidder $i$ can clinch an item and $d_{i}= D_{i}^{+}\ \forall i\in I$. Only now we can raise the price.
The preference sets, the vector of the number of unsold slots, and the set of unsold items are updated after every time a bidder clinches.
The idea of the auction is to sell slots at the highest possible price such that all slots are sold and there exists no competition between
bidders.  On the contrary, the existence of a trading path indicates that there exists competition on the 
assignment of the first slot in the path. Hence, the auction contains no trading path in the final allocation.

\begin{thm} \label{thm:main}
The allocation $(H^*,p^*)$ produced by Algorithm~\ref{alg:combclinch} is incentive compatible, individually rational, and Pareto optimal.
\end{thm}

\section{Impossibility result for diminishing marginal valuations}
\label{sec:impossibility}

In this scenario every valuation $v_{i}$ is a diminishing marginal valuation function, shaped on value $v_{i}(q)\::\:\left\{ 1,\ldots, |R| \right\} \rightarrow \mathbb{R}_{+}$.  The initial Ausubel clinching auction \cite{Ausubel04} was indeed proposed for the case of diminishing marginal utilities without budgets. Since we consider only the case of indivisible items in this section, marginal values are constant in every unit interval.

We show  that even in the case of two identical items and public budgets,  it is impossible to design an auction that is incentive compatible and Pareto optimal if marginal valuations are private. The problem is that agents can lie over their own marginal values in order to raise the price paid by the other bidders. This results in a decrease of the budget and of the demand of the other bidders, and thus a possible decrease of the price charged  to the non-truth telling bidder.

\begin{thm}\label{thm:impossibility2}
There is no incentive compatible, Pareto optimal, individually rational multi-unit auction for private diminishing  marginal valuations and public budgets.
\end{thm}

\newpage
\bibliographystyle{plain}

\newpage
\appendix
\section{Proof of Theorem~1}

Let us assume that we have a legal allocation $(X',p')$ that is Pareto superior to $(X,p)$. The utility of the auctioneer does not decrease. Thus, the sum of the payments of the bidders fulfills $\sum_{i\in I}p'_i\geq \sum_{i\in I}p_i$. If $\sum_{i\in I}p'_i> \sum_{i\in I}p_i$ then an allocation $(X',p'')$ where $\sum_{i\in I}p''_i= \sum_{i\in I}p_i$ exists, which is Pareto superior compared to $(X,p)$ as well: simply give the additional payments back to some of the bidders.
Therefore, it suffices to consider the case where $\sum_{i\in I}p'_i= \sum_{i\in I}p_i$. 

Let $q_i = \sum_{j\in J} \alpha_j (x'_{i,j} - x_{i,j})$ be the weighted capacity change of bidder $i$.
Since $(X,p)$ and $(X',p')$ are legal allocations, $\sum_{i \in I} x_{i,j}=1$ for all $j \in J$, and $\sum_{i\in I} x'_{i,j}=1$ for all $j \in J$. Hence,
$\sum_{i \in I} q_i = \sum_{i \in I} \sum_{j\in J} \alpha_j (x'_{i,j} - x_{i,j}) =   \sum_{j \in J}\alpha_j ( \sum_{i\in I} x'_{i,j} - \sum_{i \in I} x_{i,j} ) = 0$.
It follows that
(a) $\sum_{b\in I: q_b < 0} (-q_b) = \sum_{i\in I: q_i \geq 0} q_i$.
As $\sum_{i \in I} p_i = \sum_{i \in I} p'_i$ it also follows that
(b) $\sum_{b\in I: q_b < 0} (p_b - p'_b) = \sum_{i\in I: q_i \geq 0} (p'_i - p_i).$

We partition the bidders into the following three sets:
$I^- = \{b \in I| q_b < 0\}$, $B^+ = \{i \in B| q_i \geq 0\}$, and $C^+ = \{i \in I\setminus B | q_i \geq 0\} = I \setminus
(I^- \cup B^+).$
We will show below that
\begin{itemize}
\item[(A)] $\sum_{b \in I^-} (p_b - p'_b) \geq \sum_{b \in I^-} (-q_b v_b) \geq \sum_{i \in B^+} q_i v_i$,
\item[(B)] $\sum_{i \in B^+} q_i v_i \geq \sum_{i \in B^+} (p'_i - p_i),$ and
\item[(C)] $C^+ = \emptyset$.
\end{itemize}

Since $\sum_{i \in C^+} (p'_i - p_i) \leq 0$, (b) implies that
$\sum_{b \in I^-} (p_b - p'_b) \leq \sum_{i \in B^+} (p'_i - p_i).$ Combined with (A) and (B) it follows that
$\sum_{b \in I^-} (p_b - p'_b) = \sum_{i \in B^+} (p'_i - p_i)$ and that all the inequalities in (A) and (B) are
actually equations, specifically (c) $\sum_{b \in I^-} (-q_b v_b) = \sum_{i \in B^+} q_i v_i$.
Furthermore, (A) implies that the total change in utility (comparing $(X,p)$ to $(X',p')$) for all bidders $b \in I^-$, 
which is $\sum_{b \in I^-} (q_b v_b - p'_b + p_b)$, equals $0$, and
(B) implies that the change in utility for all bidders $i \in B^+$, which is
$\sum_{i \in B^+} (q_i v_i - p'_i + p_i)$, equals $0$. 
Since $C^+= \emptyset$, this implies that the total change in utility for all bidders is zero. Since the utility of the auctioneer in $(X,p)$ and in $(X',p')$ does not change either this give a contradiction to the assumption that $(X',p')$ is Pareto superior to $(X,p)$ and completes the proof of Theorem~\ref{lem:char1}.

To show (B) note that the increase in payment $p'_i - p_i$ for a bidder $i\in B$ with $q_i \geq 0$ is at
most $q_i v_i$, otherwise the utility of the bidder would drop. This shows (B). 
To show the first inequality in (A) note that the total drop in payments by a bidder $b\in I$ with $q_b < 0$ is at least $-q_b v_b$. Thus,
 $\sum_{b\in I^-} (p_b - p'_b) \geq \sum_{b\in I^-} (-q_b v_b)$.

To show the second inequality in (A)
we first show the following claims. Let $s = |B^+|$
and let $r(1), r(2), \dots, r(s)$ be an ordering of the bidders in $B^+$ in increasing order of $l(\cdot)$ such that two bidders $i$ and $i'$ with
$l(i)=l(i')$ are ordered by increasing $v$-value. We show first that $r$-ordering orders the bidders by valuation.

\begin{clm}\label{claim1}
For $1 \leq l < s$ it holds that
$v_{r(l)} \leq v_{r(l+1)}$.
\end{clm}
\begin{proof}
Assume by contradiction that $v_{r(l)} > v_{r(l+1)}$ for some $1 \leq l < s$. Since $l(r(l)) < l(r(l+1)) \leq h(r(l+1))$,
$r(l+1) \in \tilde N_{r(l)}.$ Since $r(l) \in B$, it follows that $v_{r(l)} \leq v_{r(l+1)}$. Contradiction!
\end{proof}

Note that $\tilde N_{r(l+1)} \subseteq \tilde N_{r(l)}$ for $1 \leq l < s$, i.e., a bidder $b\in I$ can belong to multiple such sets.
We define for each bidder $b \in I^-$ a unique ``top'' $i \in B^+$ to whose set $\tilde N_i$ bidder $b$ belongs. More formally, we define a mapping as follows: 
Let $p(b):= \arg \max_{i\in B^+ : b\in \tilde{N}_i} r(i)$ which is the maximum $i \in B^+$ (in $r$-order) with $b \in \tilde N_i$.
Let $A_i = \sum_{b \in I^- \cap \tilde N_i:p(b) = i} (-q_b)$.
By the definition of the mapping $p$ we have that  (d) 
$\sum_{b \in I^- \cap \tilde N_{r(1)} } (-q_b v_b) \geq$ $ \sum_{b \in I^- \cap \tilde N_{r(1)}} (-q_b v_{p(b)}) $ 
$= \sum_{i \in B^+}  (\sum_{b \in I^- \cap \tilde N_i:p(b) = i} (-q_b) ) v_i $ 
$= \sum_{i \in B^+} A_i v_i$.

The following claim simply states that all bidders from $r(l)$ to $r(s)$ ``receive'' all their increases in weighted capacity from bidders in $\tilde N_{r(l)}$.

\begin{clm}\label{claim2}
For all $1 \leq l \leq s$ it holds that $\sum_{l \leq t \leq s} A_{r(t)} = \sum_{b \in I^- \cap \tilde N_{r(l)}} (-q_b)
\geq \sum_{l \leq t \leq s} q_{r(t)}$.
\end{clm}
\begin{proof}
Consider bidders $\{ r(t)| t \in \{ l,\dots, s \}\}\subseteq B^+$. We show that bidder $i$ can increase his weighted capacity in the Pareto superior assignment $X'$ only at expenses of the  reduction of the weighted capacity of  bidders in $\tilde N_{i}$.  This in turn implies 
$\sum_{b \in I^- \cap \tilde N_{r(l)}} (-q_b) \geq \sum_{l \leq t \leq s} q_{r(t)}$. 

Let us describe the assignment $X$ and the Pareto superior assignment $X'$ by a weighted bipartite directed graph $G=(V,E\cup E')$ with the vertex set $V= I\cup J$, the edge 
sets $E=\{(i,j)\in I\times J|x_{i,j}>0\}$ and $E'= \{(j,i)\in J\times I| x'_{i,j}>0\}$, and the weights $w_{i,j}=x_{i,j}\ \forall (i,j)\in I\times J$ and $w_{j,i}=x'_{i,j}\ \forall (j,i)\in J\times I$.   Edges from $I$ to $J$ are weighted by the corresponding real-numbered  value $x_{i,j}$.  Edges from $J$ to $I$ are weighted by the corresponding real-numbered value $x'_{i,j}$.   
Consider a path $\pi = (i_1, j_1, i_2, j_2, \ldots, i_{k-1}, j_{k-1}, i_k)$ in the bipartite graph. We say that the path $\pi$ is an alternating path of length $k$ with respect to the assignments $X$ and $X'$  if $(i_t,j_t) \in E$ and $(j_t,i_{t+1}) \in E'$ for all $1 \leq t<k$.  It is an alternating cycle if $i_1 = i_k$.
Since for any assignment $\sum_{i\in I} x_{i,j} = 1\ \forall j\in J$, and $\sum_{j\in J} x_{i,j} = \kappa_i\ \forall i\in I$, it holds that
\begin{align}
\sum_{j\in J} (w_{i,j} - w_{j,i}) &= 0\ \forall i\in I,\ \text{and} \label{equalbid}\\
\sum_{i\in I} (w_{i,j} - w_{j,i}) &= 0\ \forall j\in J. \label{equalround}
\end{align}
 
We decompose the bipartite graph in a set of at most $|I|\,|J|$ alternating cycles that we denote by $\Pi$.  We start from the edge $(i,j)$ or $(j,i)$ with the lowest weight 
$\lambda = \min_{(x,y)\in E\cup E'} w_{x,y}$.  We traverse the bipartite graph starting from edge $(x,y)$ and find a path going from vertex $y$ to vertex $x$. This gives us a cycle~$\pi$. If such a path would not exist we could partition the set of vertices into three disjoint subsets: $V_1$ contains $x$ and all the start vertices of paths ending at~$x$, $V_2$ contains $y$ and all the end vertices of a paths starting at~$y$, and $V_3$ contains all the remaining vertices. The edge $(x,y)$ would be directed from a vertex in $V_1$ to a vertex that is not in $V_1$ and has a positive weight and no edge would be directed from a vertex that is not in $V_1$ to a vertex in $V_1$. Thus, $\sum_{u\in V_1, v\in V_2\cup V_3} w_{u,v}>0$ and $\sum_{u\in V_1, v\in V_2\cup V_3} w_{v,u} = 0$, which would contradict \eqref{equalbid} and \eqref{equalround}, and hence, a cycle $\pi$ has to exist.

Let us denote by $\lambda_{\pi} = \lambda$ the capacity of cycle $\pi$. We then reduce by $\lambda_{\pi}$ the weight of all edges on $\pi$ and we remove from the bipartite graph all edges with $0$ remaining weight.  Observe that equations 
\eqref{equalbid} and \eqref{equalround} still hold for the resulting graph. It is therefore possible to continue this procedure until the graph is empty. 

Given a cycle $\pi = (i_1, j_1, i_2, j_2, \ldots, i_{k-1}, j_{k-1}, i_k)$, we abuse notation by  denoting by $\pi$ also the set of bidders $\{i_1,  i_2, \ldots, i_k\}$. 
For a bidder $i\in \pi$, let us define $t_\pi(i)$ and $t'_\pi(i)$ such that $(i,t_\pi(i))\in E$ and  $(t'_\pi(i),i)\in E'$ are edges of the cycle.
We use $\alpha(j)$ for $\alpha_j$, which is the quality of slot~$j\in J$.

Given a bidder $i\in I$ and a set of alternating cycles $\Pi' \subseteq \Pi$ we define
\begin{equation*}
q_{i} (\Pi') = \sum_{\pi\in \Pi': i\in \pi} \lambda_{\pi} (\alpha(t'_\pi(i)) - \alpha(t_\pi(i)))
\end{equation*}
as the increase of the  weighted capacity of bidder $i$ when moving from the assignment $X$ to a new assignment by the    
set of cycles $\Pi'$.  Note that $q_i = q_i(\Pi)$ for every bidder $i$. It holds for each $\pi \in \Pi$ that
\begin{equation}
\label{eq:sum0}
\sum_{i\in \pi} q_i(\{\pi\}) =0.
\end{equation}

We prove the claim now by induction on a set of cycles  $\Pi$.  We actually prove the stronger statement
\begin{equation*}
\sum_{b \in \tilde N_{r(l)}\setminus B^+: q_b (\Pi) < 0} (-q_b (\Pi)) \geq  \sum_{i \in \tilde N_{r(l)}\setminus B^+: q_{i}(\Pi) \geq 0} q_{i} (\Pi)   +  \sum_{l \leq t \leq s} q_{r(t)}(\Pi).
\end{equation*}

Observe that the statement above can greatly be simplified by observing that all bidders in $\tilde N_{r(l)}$ appear in the above inequality.  
It is therefore enough to prove  for each set of cycles $\Pi$ that
\begin{equation}
\label{eq:sumless0}
\sum_{i\in \tilde N_{r(l)}} q_i(\Pi)\leq 0.
\end{equation}
It clearly holds for $\Pi = \emptyset$. 
Assume it holds for $\Pi$, we prove in the following that it then holds also for $\Pi' = \Pi \cup \{\pi \}$. 

Since 
\begin{equation}
\sum_{i\in \tilde N_{r(l)}} q_i(\Pi') = \sum_{i\in \tilde N_{r(l)}} q_i(\Pi) + \sum_{i\in \tilde N_{r(l)} \cap \pi}q_i(\{\pi\}),
\end{equation}
it is sufficient to prove
\begin{equation}
\label{eq:finalstat}
\sum_{i\in \tilde N_{r(l)} \cap \pi}q_i(\{\pi\})\leq 0.
\end{equation}

For any bidder $s\in \tilde N_{r(l)}$ and for any bidder  $i\notin \tilde N_{r(l)}$ it holds that $h(i) \leq l(s)$.  This implies in turn that 
any bidder in $i\in \pi \cap \tilde N_{r(l)}$ will only increase his weighted capacity when swapping a fraction of a slot against a fraction of a slot that is assigned to another bidder $s\in  \pi \cap \tilde N_{r(l)}$ in $X$.  
It follows
$\sum_{i\in \pi \setminus \tilde N_{r(l)}} q_i(\{\pi\}) \geq 0.$
Combined with  Equation \eqref{eq:sum0} this yields the proof of the statement of Equation \eqref{eq:finalstat}.
\end{proof}
We need one more auxiliary lemma before completing the proof of the second inequality of (B).

\begin{clm}\label{claim3}
If $(X',p')$ is a Pareto superior solution to $(X,p)$ then
for every $1 \leq l \leq s$
\begin{equation*}
\sum_{l \leq t \leq s} A_{r(t)} v_{r(t)} \geq \sum_{l \leq t \leq s} q_{r(t)} v_{r(t)} + \sum_{l \leq t \leq s}
(A_{r(t)} - q_{r(t)}) v_{r(l)}.
\end{equation*}
\end{clm}
\begin{proof}
We use backwards induction on $l$.
For $l = s$, it trivially holds that $A_{r(s)} v_{r(s)} \geq q_{r(s)} v_{r(s)} + (A_{r(s)} - q_{r(s)} )v_{r(s)}$.

For $l < s$, we use the inductive claim for $l+1$, Claim~\ref{claim2}, and the fact that $v_{r(l+1)} \geq v_{r(l)}$ according to Claim~\ref{claim1}.
Thus,
$\sum_{l \leq t \leq s} A_{r(t)} v_{r(t)} = $
$\sum_{l+1 \leq t \leq s} A_{r(t)} v_{r(t)} + A_{r(l)} v_{r(l)} \geq$
$\sum_{l+1 \leq t \leq s} q_{r(t)} v_{r(t)} + \sum_{l+1 \leq t \leq s}
(A_{r(t)} - q_{r(t)}) v_{r(l+1)} + A_{r(l)} v_{r(l)} \geq $
$\sum_{l \leq t \leq s} q_{r(t)} v_{r(t)} + $
$\sum_{l+1 \leq t \leq s}(A_{r(t)} - q_{r(t)}) v_{r(l)} + $
$(A_{r(l)} - q_{r(l)}) v_{r(l)} =$
$\sum_{l \leq t \leq s} q_{r(t)} v_{r(t)} + \sum_{l \leq t \leq s}(A_{r(t)} - q_{r(t)}) v_{r(l)}.$
\end{proof}

By Claim~\ref{claim2} it follows that $\sum_{1 \leq t \leq s} (A_{r(t)} - q_{r(t)}) \geq 0$, and thus,  by (d) and Claim~\ref{claim3} it follows that
$\sum_{b \in I^-} (-q_b) v_b \geq$
$\sum_{i \in  B^+} A_i v_i  =$
$\sum_{1 \leq t \leq s} A_{r(t)} v_{r(t)} $
$ \geq \sum_{1 \leq t \leq s} q_{r(t)} v_{r(t)} $ 
$= \sum_{i \in B^+} q_i v_i$.
This completes the proof of the second inequality of (B).

To show (C) assume by contradiction that $C^+ \not= \emptyset$ and consider two cases that follow from Claim~\ref{claim2}:

{\bf Case 1:} $\sum_{i\in B^+} A_{i} > \sum_{i \in B^+} q_i.$
Combined with (d) and Claim~\ref{claim3} this shows that
$\sum_{b \in I^-} (-q_b v_b) \geq \sum_{i\in B^+} A_{i} v_{i} > \sum_{i\in B^+} q_{i} v_{i}.$ But this is a contradiction to (c) above.

{\bf Case 2:} $\sum_{i\in B^+} A_{i} = \sum_{i \in B^+} q_i.$
Note that $\sum_{i\in B^+} A_{i} = \sum_{b \in I^- \cap \tilde N_{r(1)}} (-q_b)$.
Then (a) implies $\sum_{b \in I^- \setminus \tilde N_{r(1)}} (-q_b) = \sum_{i \in C^+} q_i > 0$.
By (c) $\sum_{i \in B^+} q_i v_i = \sum_{b \in I^-} (-q_b v_b) =$
$ \sum_{b \in I^- \cap \tilde N_{r(1)}} (-q_b v_b) +$
$\sum_{b \in I^- \setminus \tilde N_{r(1)}} (-q_b v_b) >$
$ \sum_{b \in I^- \cap \tilde N_{r(1)}} (-q_b v_b).$
By Claim~\ref{claim2}, Claim~\ref{claim3}, and (d)  it follows  that \linebreak
$\sum_{b \in I^- \cap \tilde N_{r(1)}} (-q_b v_b) \geq \sum_{i \in B^+} A_i v_i \geq \sum_{i \in B^+} q_i v_i$, which contradicts the previous statement.

\section{Proof of Theorem~3}

Let $(X,p)$ be the allocation computed by the auction and assume by contradiction that there exists a trading swap, i.e., a sequence of bidders $(u = a_0, a_1, \dots, a_k = w)$
that fulfills the above conditions.
Consider the Pareto superior allocation $(X', p')$ constructed in the proof of Theorem~\ref{lem:char2}. Define $\cfin_i := \sum_{j\in J} \alpha_j x_{i,j}$ and $c'_i := \sum_{j \in J} \alpha_j x'_{i,j}$ for all bidders~$i$.
Note that $c'_w = \cfin_w - \delta$, $c'_u = \cfin_u + \delta$, and $\cfin_i = c'_i\ \forall i\in I\setminus \{ u,w\}$. Let $\delta' =  \delta \frac{\pi}{\pi^{+}}  > 0$.

We construct a modified Pareto superior allocation $(X'',p'')$ with $c''_w = \cfin_w - \delta'$, $c''_u = \cfin_u + \delta'$, and $c''_i = \cfin_i\ \forall i\in I\setminus\{ u,w\}$, where 
$c''_i = \sum_{j\in J} \alpha_j x''_{i,j}$.
Specifically, we use the same set of bidders $u = a_0, \dots,a_k = w$, perform the swaps between the same bidders as for $(X',p')$, but use as
swap values $\tau_{p+1}' := \tau_{p+1} \frac{\pi}{\pi^{+}}$ instead of $\tau_{p+1}$ and as payments
$p_{u}'' = p_{u} + v_{w} \delta'$, $p_w'' = p_w - v_w \delta'$, and $p_i'' = p_i$ for all other bidders~$i$.
By the same argument as for $(X', p')$ the allocation $(X'',p'')$ is Pareto superior to $(X,p)$.

We will show that $(X'',p'')$ can be used to construct a smaller feasible solution to one of the linear programs solved by \textsc{Sell}. Since the linear program has found the minimal solution this leads to a contradiction with the assumption that there exists a trading swap in $(X,p)$.

Let $\bfin_i := b_i - p_i$ be the remaining budget of bidder~$i$ at the end of the algorithm.
The value $c_w$ of bidder $w$ increases only when procedure \textsc{Sell} returns a non-zero value for $\gamma_{w}$, where $w$ was the last parameter when \textsc{Sell} was called, that is, the linear program solved in \textsc{Sell} was trying to minimize $\gamma_{w}$. Since $\cfin_w > c_w''$, there exists
a unique call to procedure \textsc{Sell} 
with parameters $(I,J, \alpha,\kappa, v, c, d, w)$
such that {\em before} the execution of the linear program $c_w \leq c''_w$ and
\textsc{Sell} returns a value $s>0$ such that $c_w + s > c''_w$.
We call the corresponding linear program LP. Its inputs are the vectors $c$, $d$, and $\kappa$, its variables are the matrix $X = (x_{i,j})_{(i,j) \in I \times J}$ and the vector $\gamma$.
 Let $\pi$ be the price at the time of the call. 

We will show that using $(X'', p'')$ we can construct a feasible solution for this  linear program which outputs a value $s' < s$. This leads to the desired contradiction.

We first show the following claim:

\begin{clm} 
Using $(X,p)$ we can find a feasible solution $(X,\tilde{\gamma})$ to LP that fulfills the following additional conditions:
\begin{itemize}
\item for all bidders $i\in A\setminus E$ with $i \not= w$ and $d_i > d_i^+$:
$\tilde{\gamma}_i \leq d_i - \frac{\bfin_i}{\pi}$
\item for all bidders $i\in A\setminus E$ with $i \not= w$ and $d_i = d_i^+$:
$\tilde{\gamma}_i \leq d_i - \frac{\bfin_i}{\pi^{+}}$
\end{itemize}
\end{clm}
\begin{proof}
We first recall  that the case $d_i > d_i^+$ happens when the LP is computed before the demand of bidder $i$ has been updated as in line 25 of \textsc{Auction}.  The case  $d_i = d_i^+$ happens after the update has been made for bidder $i$.  

To prove this claim we
set $\tilde{\gamma}_i := \sum_{j\in J} \alpha_j x_{i,j} - c_i = \cfin_i - c_i$.  
First we show that $(X, \tilde{\gamma})$ fulfills the constraints of LP.
Since the allocation $(X,p)$ is derived from the last linear program executed by the algorithm, 
it fulfills the conditions
$\sum_{i \in I} x_{i,j} = 1\ \forall j \in J$ and $\sum_{j \in J} x_{i,j} = \kappa_i\ \forall i \in I$. 
By definition $\sum_{j \in J} x_{i,j} \alpha_j - \tilde{\gamma}_i = c_i\ \forall i \in I$.

 Recall that
$\bfin_i$ is the remaining budget of bidder~$i$ at the end of the auction, that is, the money {\em not} spent by $i$.
Note that  bidder~$i$ clinched $\tilde{\gamma}_i = \cfin_i - c_i$ ``weighted capacity'' after LP was executed.

{\bf Case 1:} Consider first a bidder $i \not= w$ with $d_i > d_i^+$.
Note that for bidders of this type the remaining budget when LP is called is
$d_i \pi$ and that these bidders pay a price per ``weighted capacity unit'' of at least $\pi$ for all capacity that was not clinched before LP was executed. Thus, bidder~$i$ pays $ d_i \pi - \bfin_i$ for all the ``weighted capacity'' that was not clinched before LP was executed. 
Thus, $\tilde{\gamma}_i \pi \leq d_i \pi - \bfin_i$.

{\bf Case 2:} Consider next a bidder $i \not= w$ with $d_i = d_i^+$. 
Note that for bidders of this type the remaining budget when LP is called is
$d_i \pi^{+}$ and that these bidders pay a price per ``weighted capacity unit'' of at least $\pi^{+}$ for all capacity that was not clinched before LP since they can only clinch 
at the price $\pi^{+}$ or higher. Note that we know that $\pi^{+} \leq v_i:$ Since $i \in A\setminus E$ it holds that $v_i > \pi$, and therefore, $v_i\geq \pi^{+}$. 
Thus, bidder~$i$ pays $d_i \pi^{+} - \bfin_i$ for all the ``weighted capacity'' clinched after LP was executed. 
Thus, $\tilde{\gamma}_i \pi^{+} \leq d_i \pi^{+} - \bfin_i.$
\end{proof}

Next we define $\gamma_i'' = \sum_{j\in J} x''_{i,j} \alpha_j - c_i = c_i'' - c_i$ for all $i\in I$ and show that $(X'', \gamma'')$ is a feasible solution of LP and that $\gamma''_w < s$ thus leading to a contradiction.
Note that $\gamma_u'' = \tilde{\gamma}_u + \delta'$.
By the definition of $X''$ for all $i \in I$ it holds that $\sum_{j\in J} x''_{i,j} = \sum_{j\in J} x_{i,j} = \kappa_i$ as
whenever for some $\tau$ with $-1\leq \tau\leq 1$, $x_{i,j}''$ is set to $x_{i,j} + \tau$ for some $j \in J$, $x_{i,l}''$ is set to $x_{i,l} - \tau$ for some other $l \in J$.
Additionally for every $j \in J$ it holds that $\sum_{i\in I} x_{i,j}'' = \sum_{i\in I} x_{i,j} = 1$ as whenever 
$x_{a_p,j}''$ is set to $x_{a_p,j} + \tau$ for some $\tau$ with $-1\leq\tau\leq 1$, either $x_{a_{p+1}, j}''$ is set to $x_{a_{p+1}, j} - \tau$, or $x_{a_{p-1}, j}''$ is set to $x_{a_{p-1}, j} - \tau$.
Thus $(X'', \gamma'')$ fulfills constraints (a) and (b) of LP. By the definition of $\gamma''$
constraint (c) also holds. 

For constraint (d)  note that for all $i \in I \setminus \{u, w\}$ we know that $\gamma''_i = \tilde{\gamma}_i \leq d_i$, and thus, constraint (d) holds for such $i$. 
For $i = w$,  by definition of a trading swap $\sum_{j\in J} \alpha_j  x''_{i,j} < \sum_{j\in J} \alpha_j x_{i,j}$, and thus,
$\gamma''_w < \tilde{\gamma}_w\leq d_w$. Hence constraint (d) also holds for $i = w$.
For $i = u$,  we know that $\gamma''_u = \tilde{\gamma}_u + \delta'$ and we have to show that $d_u \geq \gamma_u''$.

Since $\cfin_w > c_w$ we know that $w$ is still an active bidder when LP is executed, and thus,
$v_w \geq \pi$. 
Hence, $\bfin_u = b_u - p_u \geq p'_u - p_u = v_w \delta \geq \pi \delta$. 

By $v_w\geq\pi$ it follows from the definition of a trading swap that $v_u > \pi$ and that therefore $u\in A\setminus E$.
Consider first the case that $d_u > d_u^+$.
By the previous claim it follows that
$d_u \geq \tilde{\gamma}_u + \frac{\bfin_u}{\pi} \geq \tilde{\gamma}_u + \delta  = \gamma_u'' + \delta - \delta' >  \gamma_u''$.
Consider next the case that $d_u = d_u^+$. By 
 the previous claim it follows that
$d_u \geq \tilde{\gamma}_u + \frac{\bfin_u}{\pi^{+}} \geq \tilde{\gamma}_u + \delta \frac{\pi}{\pi^{+}} = \tilde{\gamma}_u + \delta'  = \gamma_u''$.

It remains to show that $\gamma''_w < s$. Recall that by the definition of LP it holds that
$c_w + s > c_w''$, while, by definition of $\gamma_w''$, $c_w + \gamma_w'' = c_w''$. Thus $\gamma_w'' < s$, which leads to the desired contradiction.

\section{Proof of Lemma~1}

We first show the following claim:
\begin{clm}\label{claim:randpo}
For every legal allocation $(N,p)$ in the indivisible case there exists a legal allocation $(X,p)$ in the  divisible case where all the bidders and the auctioneer have the same utility.
\end{clm}
\begin{proof}
The utility of the auctioneer stays unchanged, since we leave the payments unchanged. We set $x_{i,j}=\frac{|\{ r\in R| n_{j,r} = i\}|}{|R|}\ \forall i\in I,\ \forall j\in J$. The utility of bidder $i$ is the same for $(N,p)$ and $(X,p)$, since the utility of bidder $i$ is $\sum_{j\in J}\frac{\alpha_j}{|R|}|\{ r\in R | n_{j,r} = i \}| v_i -p_i=\sum_{j\in J} \alpha_j x_{i,j} v_i -p_i$ for $(N,p)$. The demand constraint for $(N,p)$ implies $\kappa_i\geq \max_{r\in R}|\{ j\in J| n_{j,r}=i\}|\geq \frac{|\{ (j,r)\in J\times R | n_{j,r} = i\}|}{|R|} = \sum_{j\in J}\frac{|\{ r\in R |n_{j,r}=i\}|}{|R|}=\sum_{j\in J}x_{i,j}$, and therefore it implies the demand constraint in $(X,p)$. Since all the slots are fully assigned to the bidders in $(N,p)$, and consequently for $(X,p)$, it follows that $(X,p)$ is legal. 
\end{proof}
Given a probability distribution over legal allocations for the indivisible case, transform each legal allocation
that has a non-zero probability into a legal allocation for the divisible case. Then create a new allocation for the divisible case by adding up the all of these legal allocations for the divisible case weighted by the probability distribution. Since the weights are created by a probability distribution, they add up to 1, and thus, the resulting combined allocation fulfills Conditions (1) and (2) of a legal allocation. As the payment is identical to the payment
for the indivisible case, Condition (3) is also fulfilled.

\section{Construction of the matrices in Section~4}

We construct $\tilde{M}$ in the following way: we give every entry $(j,c)\in J\times C$ of $M$ that has the value $i\in I$ a unique number $l_{j,c}$ that is starting from 1; we assign the same value to two entries $(j,c)$ and $(j',c')$ in $\tilde{M}$ if and only if the corresponding entries in $M$ have the same values ($m_{j,c}=m_{j',c'}$) and their numbers $l_{j,c}$ and $l_{j',c'}$ fulfill $\lfloor\frac{l_{j,c}}{|C|}\rfloor=\lfloor\frac{l_{j',c'}}{|C|}\rfloor$. To be more concrete, let the indicator variable $\delta_{i,j,c}=1$ if $m_{j,c}=i$ and  $\delta_{i,j,c}=0$ otherwise.  We define $l_{j,c}:=\sum_{(j',c')\in S_{j,c}} \delta_{m_{j,c},j',c'}$ where $S_{j,c}=\{ (j',c')\in J\times C | (j'<j)\vee (j'=j \wedge c'\leq c)\}$ and construct matrix $\tilde{M}$ of size $|J|\times |C|$ by $\tilde{m}_{j,c}=\left( \sum_{i'=1}^{m_{j,c}-1}\kappa_{i'}\right) + \lfloor \frac{l_{j,c}}{|C|}\rfloor +1$. The matrix $\tilde{M}$ has the property that all the entries that have an identical value in $\tilde M$ have an identical value in $M$, but every value appears at most $|C|$ times in $\tilde M$. Theorem~\ref{thm:swap} gives the swapping algorithm that is applied to $\tilde M$. The pseudo-code of the algorithm is presented in the proof of the theorem.

We can now use the swapping algorithm on $\tilde{M}$ and get a matrix $\tilde{M}'$ of size $|J|\times|C|$ where none value appears more than once in the same column. 

To construct matrix $M'$ we simply reproduce the swaps that happened to matrix $\tilde{M}$ on matrix $M$. We define the matrix $M'$ with size $|J|\times |C|$ where $m'_{j,c}=\min\{ v\in I | \sum_{i=1}^{v-1}\kappa_{i}< m_{j,c}\}$. The values of the entries of $M'$ correspond to the bidders in $I$ and in each column each value $i\in I$ appears at most $\kappa_i$ times. This methodology preserves the amount of capacity of each slot $j\in J$ that is allocated to each bidder $i\in I$ (i.e., $|\{c\in C| m'_{j,c} = i\}| = \sum_{c\in C} \delta_{i,j,c} = y_{i,j}\ \forall i\in I, \forall j\in J$).

\section{Proof of Theorem~4}

Our goal is to find an algorithm that swaps the values between the entries such that each value appears only once in each column. We define $r=|J|$ and $n=|C|$. 
Let  the indicator variable $\delta_{i,j,c}=1$ if $m_{j,c}=i$,  $\delta_{i,j,c}=0$ otherwise.
We define the badness of a value $i\in I$  as $\beta_i(M) = \sum_{j\in J} \sum_{c\in C} (\delta_{i,j,c}) - |\{c\in C| \sum_{j\in J}  \delta_{i,j,c} >0\}|$ 
as the difference between the number of entries which have value $i$ and the number of columns in which $i$ appears. 
Moreover, we define by $\beta(M)=\sum_{i\in I} \beta_i(M)$ the badness  of matrix $M$.  When each value appears at most once in each column, the badness of the matrix is $0$. We aim at reducing the badness of the matrix at each sequence of swaps. 

Let us assume that $i$ appears more than once in column $a$. Then, there exists a column $b$ where $i$ does not appear at all, because each value appears in at most $n$ entries. For the following operations we consider only the columns $a$ and $b$.  We now define a sequence of swaps between pairs of entries of the two columns.  We can see the two columns as the two sides of a bipartite graph.  We set vertices $A =\{a_1,\ldots, a_r\}$ on the left side and vertices $B =\{b_1,\ldots, b_r\}$ on the right side.  The values of the $a_j$ and $b_j$ are $m_{j, a}$ and $m_{j, b}$, respectively  for all $j\in J$.
We set edges $\{(a_j,b_j)|j=1,\ldots, r\}$ from left vertices to right vertices of the same slot, and edges $\{(b_j,a_k)| m_{j,b} = m_{k,a}\}$ from right vertices to left vertices with same value. 

We define a swapping alternating path $(a_{j_1}, b_{j_1}, \ldots, a_{j_t}, b_{j_t})$ on the bipartite graph.  The path starts with a vertex of the left side and ends with a vertex of the right side.  We start with  $a_{j_1}$, one of the vertices on the left side with value $i$, and set $i_0=i$. Vertex $a_{j_k}$ is followed in the path by vertex $b_{j_k}$.  Let $m_{j_k,b}=i_k$ be the value of vertex $b_{j_k}$.  If value $i_k$ appears more than once on column $b$ then  we end the path.  Otherwise, we continue to the left with any of the edges $(b_{j_k},a_{j_{k+1}})$, if any such edge exists. Finally, we implement $t$ swaps by exchanging the values of the endpoints  $(a_{j_k}, b_{j_k})$ of each edge on the path, i.e., we exchange $m_{j_k,a}$ with $m_{j_k,b}$.

We prove two claims:

\begin{clm}
The sequence $(i_1,\ldots,i_t)$ does not contain any value more than once and it does not contain value~$i_0$.
\end{clm} 
\begin{proof}
Assume that the path has reached vertex $b_{j_k}$ on the right side of the graph.  The path continues to the left side only if the value $i_k$ 
appears only once on the right side.  Therefore, the sequence $(i_1,\ldots,i_t)$ contains the values of the vertices on the right side only once.
The set does not contain $i_0=i$ since value $i$ does not appear on $b$.  
\end{proof}

\begin{clm}
The sequence of swaps along the edges $(a_{j_k}, b_{j_k})$, $k=1,\ldots,t$,  reduces the total badness of the matrix by at least 1. 
\end{clm}
\begin{proof}
The first swap of the path reduces the badness of bidder $i=i_0$  by 1, since there exists no value $i$ on any entry of column $b$. 
We now prove that the total badness of bidders $\{i_1,\ldots,i_t\}$ does not increase. Consider any value $i_k$ with $k<t$. 
If value $i_k$ is moved to entry $m_{j_k,a}$ then  $i_k$ appears only once in $b$.  However, value $i_k$ is also moved from entry $m_{j_{k+1},a}$ to $m_{j_{k+1},b}$. Thus, the badness of bidder $i_k$ does not increase. For value $i_t$, that is moved from  $m_{j_t,b}$ to  $m_{j_t,a}$,  we observe that either value $i_t$ appears more than once in $b$ or that there is no entry on $a$ that contains $i_t$. In both cases the badness of $i_t$ does not increase. 
\end{proof}

\begin{algorithm}
\caption{Swapping Algorithm.}
\label{alg:swap}
\begin{algorithmic}[1]
	\Procedure{Swapping}{$M$}
		\State $J\gets\{ 1,\dots,\mathrm{rows}(M)\}$
		\State $C\gets\{ 1,\dots,\mathrm{columns}(M)\}$
		\State $n\gets \max_{(j,c)\in J\times C} m_{j,c}$
		\For{$i\in\{1,\dots,n\}$}
		\State $j_c\gets |\{ j\in J| m_{j,c}=i\}|\ \forall c\in C$
		\While{$\max_{c\in C}(j_c)>1$}
		\State $a\gets \min(\{c\in C| j_c>1\})$
		\State $b\gets \min(\{c\in C| j_c=0\})$
		\State $i'\gets i$
		\State $k\gets 0$
		\Repeat		
		\State $k\gets \min(\{j\in J\setminus\{ k \}| m_{j,a}=i'\})$
		\State $m_{k,a}\gets m_{k,b}$
		\State $m_{k,b}\gets i'$
		\State $i'\gets m_{k,a}$
		\Until{$|\{ j\in J| m_{j,a}=i' \}| = 1 \vee |\{ j\in J| m_{j,b}=i' \}|>0$}
		\State $j_a\gets j_a-1$
		\State $j_b\gets j_b+1$
		\EndWhile
		\EndFor
		\State \Return{$M$}
	\EndProcedure
\end{algorithmic}
\end{algorithm}

The proof of the theorem is therefore completed. Algorithm~\ref{alg:swap} implements the idea of the swaps that is described above.

\section{Proof of Theorem~5}

In order to prove Theorem \ref{thm:paropt} we need the following lemmas and the definition of Pareto equivalent trading paths.

\begin{lem}\label{lem:paths-and-cycles}

Let $H$ and $H^{'}$ be two allocations with all items allocated. The symmetric difference $H \ominus H^{'}$ between the two allocations can be decomposed into a set of alternating paths (simple paths and paths with cycles) with respect to $H$.

\end{lem}
\begin{proof}
Let the graph $G$ be a directed graph. The edges from the matching $H$ are directed from bidders to items and the edges from the matching $H^{'}$ are directed from items to bidders.
Since $H$ and $H^{'}$ have all items allocated, in both matchings there are $m$ edges for every item. Moreover, each item  in $G$ will have an equal number of incoming and outgoing edges.  Thus, no item has to be the start or the end of a path, and we can always find two edges incident to any item such that there are no two consecutive edges from $H$ or $H^{'}$. 
\end{proof}

\begin{proof}[\bf Proof of Lemma \ref{lem:simple-complex}]
Assume that  $\sigma = (a_1, t_1,  \ldots, t_{j-1}, a_j)$ is an alternating path with a cycle. Let $s_k$ be the first vertex of the cycle in the order.  
We decompose the path into  $\sigma_s = (a_1, \ldots, s_{k-1}, s_k)$, $\sigma_c = (s_k,s_{k+1}, \ldots, s_{i-1}, s_k)$, and
$\sigma_e = (s_k,s_{i+1}, \ldots, a_j)$, where $s_k$ is either a bidder or an item.
If $s_k$ is an item, say  $t$, then  $s_{k-1}$, $s_{k+1}$, $s_{i-1}$, and $s_{i+1}$ are bidders, and $t \in S_{s_{k-1}}$, $t \in  S_{s_{k+1}}$, $t \in   S_{s_{i-1}}$, and $t \in   S_{s_{i+1}}$. Moreover,  $t  \in H_{s_{k-1}}$,
$t\in H_{s_{i-1}}$,
$t \not\in H_{s_{k+1}}$, and $t \not\in H_{s_{i+1}}$ by the definition of an alternating path.
Thus, the concatenation of $\sigma_s$ and $\sigma_e$ is still an alternating path and it is simple.
In the same way, if $s_k$ is a bidder, say $a$,  we have that $s_{k-1}$, $s_{k+1}$, $s_{i-1}$, and $s_{i+1}$ are items and $ \{s_{k-1}, s_{k+1}, s_{i-1}, s_{i+1}\} \subseteq S_{a}$, $s_{k-1}\not\in H_{a},
s_{i-1} \not\in H_{a}$, $s_{k+1} \in H_{a}, s_{i+1} \in H_{a}$. 
Thus, we can again concatenate $\sigma_s$ and $\sigma_e$ and obtain a simple alternating path.
The above process can be iterated if there exist more cycles in the path. 
\end{proof}

We conclude that for every trading path there is a Pareto equivalent simple trading path and that if there are no simple trading paths then there are no trading paths at all.

Now we are ready for the proof of Theorem \ref{thm:paropt}.  This proof and the proof of Theorem~\ref{thm:main} follow very closely those of \cite{FLSS11}.

\begin{proof}[\bf Proof of Theorem \ref{thm:paropt}]
Let $\mathcal{Q}$ be the predicate that $(H,p)$ is Pareto optimal, $\mathcal{R}_1$ be the predicate that all items are sold in $(H,p)$, and $\mathcal{R}_2$ the
predicate that there are no trading paths in $G$ with respect to $(H,p)$. We seek to show that $\mathcal{Q} \Leftrightarrow \mathcal{R}_1 \wedge \mathcal{R}_2$.

$\mathcal{Q} \Rightarrow (\mathcal{R}_1\wedge \mathcal{R}_2)$: to prove this we show that $(\neg \mathcal{R}_1 \vee \neg \mathcal{R}_2) \Rightarrow \neg \mathcal{Q}$.

If both $\mathcal{R}_1$ and $\mathcal{R}_2$ are true then this becomes $\mathrm{False} \Rightarrow \neg\mathcal{Q}$ which is trivially true.

If the allocation
$(H,p)$ does not assign all items ($\neg \mathcal{R}_1$) then it is clearly not Pareto optimal ($\neg \mathcal{Q}$). We can get a better allocation by assigning
unsold items to any bidder $i$ with such items in $S_i$. This increases the utility of bidder $i$.

If $\neg \mathcal{R}_2$ then there exists a trading path $\sigma$ in $G$ with respect to $(H,p)$. Let $\sigma = (a_1, t_1,$ $a_2, t_2, \ldots, a_{j-1}, t_{j-1}, a_j)$, with $v_{a_j} > v_{a_1}$
and $b^*_{a_j} \geq v_{a_1}$; then we can decrease the payment of bidder $a_1$ by $v_{a_1}$, increase the payment of bidder $a_j$ by the same $v_{a_1}$, and move item $t_i$ from bidder $a_i$ to
bidder $a_{i+1}$ for all $i=1,\ldots, j-1$. In this case, the utility of bidders $a_1, a_2, \ldots, a_{j-1}$ is unchanged, the utility of bidder $a_j$ increases by $v_{a_j}-v_{a_i} > 0$, and
the utility of the auctioneer is unchanged. The sum of payments by the bidders is likewise unchanged. This contradicts the assumption that $(H,p)$ is Pareto optimal.

We now seek to prove that $(\mathcal{R}_1 \wedge \mathcal{R}_2) \Rightarrow \mathcal{Q}$. We note above that if not all items are allocated ($\neg \mathcal{R}_1$) then the allocation is not
Pareto optimal ($\neg \mathcal{Q}$), thus $\mathcal{Q} \Rightarrow \mathcal{R}_1$ and (trivially) $\mathcal{Q} \Rightarrow \mathcal{Q} \wedge \mathcal{R}_1$ (Pareto optimality implies that all
items
allocated). Thus, $(\mathcal{R}_1 \wedge \mathcal{R}_2) \Rightarrow \mathcal{Q} \Rightarrow \mathcal{Q} \wedge \mathcal{R}_1$. If $\mathcal{R}_1$ is false this predicate becomes $\mathrm{False}\Rightarrow \textrm{False}$, thus we remain with the case where all items are allocated.

We show the contrapositive: $\neg \mathcal{Q}  \Rightarrow (\neg \mathcal{R}_1 \vee \neg \mathcal{R}_2)$.
Assume $\neg \mathcal{Q}$, i.e., assume that $(H,p)$ is not Pareto optimal.  Further assume $\mathcal{R}_1$, that $H$ assigns all items.  We will show $\neg \mathcal{R}_2$, i.e., that there is a trading path with respect to $(H,p)$.  Since $(H,p)$ is not Pareto optimal, there must be some other allocation $(H',p')$ that is not worse for all players (including the auctioneer) and strictly better for at least one player. We can assume that $(H',p')$ assigns all items as well, as otherwise we can take an even better allocation that would assign all items.

By Lemma~\ref{lem:paths-and-cycles} we know that $H$ and $H'$ are related by a set of alternating paths (simple and not) and cycles. On a path, the first bidder gives up one item, whereas the last bidder receives one item more, after items are exchanged along the path. Cycles represent giving up one item in return for another by passing items around along it. Cycles do not change the number of items assigned to the bidders along the cycles so we will ignore them.
Moreover, by Lemma \ref{lem:simple-complex} we know that every trading path that is not simple has a Pareto equivalent simple trading path and if it does not exist any simple trading path then there are no trading paths at all.
Thus, we can focus on the existence of simple trading paths.
Let us denote the number of alternating paths by $z$, and denote the start and end bidders along these $z$ alternating paths by $x_1,\ldots,x_z$ and $y_1,\ldots, y_z$. We assume that the same bidder may appear multiple times amongst $x_i$'s or multiple times amongst $y_i$'s, but cannot appear both as an $x_i$ and as a $y_i$, since we can concatenate two such paths into one. Such an alternating path represents a shuffle of items between bidders where bidder $x_j$ loses an item, and bidder $y_j$ gains an item  when  moving from $H$ to $H'$. In general, these two items may be entirely different.

Assume there are no trading paths with respect to $(H,p)$.  Then it must be the case that for each alternating path $j$ either
$v_{y_j} \leq v_{x_j}$ holds, $b^*_{y_j} < v_{x_j}$ holds, or both holds, where $b^*_{y_j}$ is the budget left over for bidder $y_j$ at the end of the mechanism.
We define $\mu = \{j\in \{ 1,\dots, z\} | v_{y_j} \leq v_{x_j}\}$ and $\nu =\{ 1,\dots, z\}\setminus \mu$.

Now, no bidder is worse off in $(H',p')$ in comparison to $(H,p)$, the auctioneer is not worse off, and, by assumption, either
\begin{enumerate}
\item[(A)] some bidder is strictly better off, or
\item[(B)] the auctioneer is strictly better off.
\end{enumerate}

First, we rule out case (B) above:
Consider the process of changing $(H,p)$ into $(H',p')$ as a two stage process: at first, the bidders $x_1, \ldots, x_z$ give up items. During this first stage, the payments made by these bidders must decrease (in sum) by at least $Z^-=\sum_{i=1}^z{v_{x_i}}$. The second stage is that bidders ${y_1,\ldots,y_z}$ receive their extra items. In the second stage, the maximum extra payment that can be received from bidders ${y_1,\ldots,y_z}$ is no more than
\begin{equation} Z^+=\sum_{j\in \mu} v_{y_j} + \sum_{j\in \nu} b^*_{y_j} \leq \sum_{j\in \mu} v_{x_j} + \sum_{j\in \nu} v_{x_j} = Z^-,\label{eq:zplezm}\end{equation}
by definition of sets $\mu$ and $\nu$ above. Thus, the total increase in revenue to the auctioneer would be $Z^+ - Z^- \leq 0$. This rules out case B. Moreover, as the auctioneer cannot be worse off, $Z^+ = Z^-$ and from Equation \eqref{eq:zplezm} we conclude that
\begin{equation}\sum_{j\in \mu} v_{y_j} + \sum_{j\in \nu} b^*_{y_j} = \sum_{j\in \mu} v_{x_j} + \sum_{j\in \nu} v_{x_j}. \label{eq:equal}\end{equation}

By definition, we have for $j\in \mu$ that $v_{y_j} \leq v_{x_j}$, and for $j\in \nu$ we have that $b^*_{y_j} < v_{x_j}$.
Thus, if $\nu \neq \emptyset$ then the left hand side of Equation \eqref{eq:equal} is strictly less than the right hand side, a contradiction. Therefore, case (A) must hold and it must be that $\nu = \emptyset$. We will conclude the proof of the theorem by showing that these two are inconsistent. By~(A), we have that
$|H'_a|  v_{a} - p'_{a} = |H_a| v_a - p_{a}$ for each bidder $a$ whose utility does not increase, $|H'_{\hat{a}}|  v_{\hat{a}} - p'_{\hat{a}} > |H_{\hat{a}}| v_{\hat{a}} - p_{\hat{a}}$ for at least one bidder $\hat{a}$, and $\sum_{a\in I} p'_a = \sum_{a\in I} p_a$.
We can now derive that
\begin{equation*}
   \sum_{a\in I} |H'_a|  v_{a}   > \sum_{a\in I} |H_a| v_a  + \left(\sum_{a\in I} p'_a - \sum_{a\in I} p_a\right),
\end{equation*}
and hence,
\begin{equation}
\sum_{a\in I} (|H'_a| - |H_a|) v_a > 0.
\label{eq:sumdiffm}
\end{equation}
  Now, whenever $a = x_j$ for a $j\in \{ 1,\dots, z\}$ we decrease $|H'_a| - |H_a|$ by one, whenever $a= y_j$ for a $j\in \{ 1,\dots, z\}$ we increase $|H'_a| - |H_a|$ by one.
  Thus, rewriting Equation~\eqref{eq:sumdiffm} we get that
  \begin{equation*}
  \sum_{a\in I}  (|\{j\in \{ 1,\dots,z\} | a=y_j\}| - |\{j\in \{ 1,\dots, z\} |a=x_j\}|) v_a > 0,
  \end{equation*}
respectively  
\begin{equation*}
  \sum_{j=1}^z v_{y_j} - \sum_{j=1}^z v_{x_j} > 0.
\end{equation*}
Hence,
\begin{equation}
  \sum_{j=1}^z v_{y_j} > \sum_{j=1}^z v_{x_j},
  \label{eq:sumj1z}
\end{equation}
but Equation~\eqref{eq:sumj1z} is inconsistent with Equation~\eqref{eq:equal} as $\nu=\emptyset$ implies that $\mu=\{1, \ldots, z\}$.
\end{proof}

\section{Proof of Theorem~6}

We present the proof of Pareto optimality for the auction described in Algorithm \ref{alg:combclinch}.

\begin{algorithm}[p]
\begin{algorithmic}[1]
\Procedure{Combinatorial Auction with Budgets}{$v,b,(S_i)_{i\in I}$}
\State{$\pi\gets 0$}
\While{$A\neq \emptyset$} \label{aucline:hitrue}
\State{\textsc{Sell}($E$)} \label{line:SellV}
\State{$A \gets A-E$} \label{line:VHi}
\Repeat \label{aucline:repstart}
\If{$\exists i | B(\neg\{i\})<\bar t$} {\textsc{Sell}($\{ i\}$)} \label{line:Selli}
\Else
\State{For an arbitrary bidder $i$ with $d_i>D^+_i(\pi)$: \\ \qquad \qquad \qquad \qquad  $d_i\gets D^+_i(\pi)$} \label{line:Hifalse}
\EndIf
\Until{$\forall i$: $(d_i=D^+_i(\pi)) \wedge (B(\neg\{i\})\geq \bar t)$} \label{aucline:repend}
\State{Increase $\pi$ until for some $i$, $D_i(\pi) \neq D^+_i(\pi)$ } \label{aucline:incprice}
\EndWhile
\EndProcedure
\end{algorithmic} \caption{Combinatorial Auction with Budgets} \label{alg:combclinch}
\end{algorithm}

\begin{algorithm}[p]
\begin{algorithmic}[1]
\Procedure{$S$-Avoid Matching}{}

\State{Construct interest graph $G$:}

\begin{itemize} 
\item Each active bidder $a\in A$ on the left with capacity constraint $d_a$.
\item Each unsold item $r \in R$ on the right with capacity constraint $c_r$.
\item Edge $(a,r)$ from bidder $a\in A$ to unsold item $r\in R$ iff $r\in S_a$.
\end{itemize}

\State{Return maximal $B$-matching with minimal number of items assigned to bidders in $S$,\newline
\mbox{\hspace{1.5em}}amongst all maximal $B$-matchings.}

\EndProcedure
\end{algorithmic} \caption{Compute an Avoid Matching via Min Cost Max Flow} \label{alg:avoidmatch}
\end{algorithm}

\begin{algorithm}[p]
\begin{algorithmic}[1]
\Procedure{Sell}{$S$}

\Repeat

\State Compute $Y = $ {\sc $S$-Avoid Matching}

\State For arbitrary $(a,r)$ in $Y$ with $a \in S$, sell item $r$ to bidder $a$ and set $S_a\gets S_a\setminus\{ r \}$.

\Until $B(\neg S) \geq \bar t$
 \EndProcedure
\end{algorithmic} \caption{Selling to the Set $S$ of Bidders} \label{alg:sellcomb}
\end{algorithm}

We first state the fact  that the auction will sell all slots of all rounds.  
As stated in Theorem \ref{thm:paropt}, this is a necessary condition for Pareto optimality.

\begin{lem}
\label{lem:sells-all}
If the multiset-union of all preference sets $S=\biguplus_{i\in I} S_i$ fulfills $U \subseteq S$, the auction will sell all items.
\end{lem}
\begin{proof}
At the beginning of the auction the price is zero, and thus, every bidder demands his hole preference set $S_i$, and all slots can be sold. During the auction the demand of bidder~$i$ decreases either when he buys a slot, or when his demand gets updated to $D_i^+$. The first case does not affect the demand for the unsold slots. In the second case, the usage of the $B$-matching guarantees that the other bidders demand all the unsold slots at the current price when we decrease $i$'s demand.
\end{proof}

Now we need to show that there are no trading paths in the final allocation $(H^*,p^*)$ produced by Algorithm~\ref{alg:combclinch}. Consider the set of all trading paths $\Sigma$ in the final allocation $(H^*,p^*)$.
\begin{Def} \label{def:suppi}
We define for every $\sigma \in \Sigma$: 
\begin{itemize}

\item Let $Y^{\sigma}$ be the $S$-avoid matching used the first time some item $r$ is sold to some bidder $a$ where $(a,r)$ is an edge along $\sigma$.
$Y^{\sigma}$ is either an $E$-avoid matching (line~\ref{line:SellV} of Algorithm~\ref{alg:combclinch}) or an $a$-avoid matching for some bidder-item
edge $(a,r)$ along $\sigma$ (line~\ref{line:Selli} of Algorithm~\ref{alg:combclinch}).

\item If $Y^{\sigma}$ is an $E$-avoid matching, let $E^{\sigma}$ be this set of exiting bidders.

\item If $Y^{\sigma}$ is an $a$-avoid matching, let $a^{\sigma}$ be this bidder.

\item Let $F^{\sigma}\subseteq H^*$ be the set of edges $(a,r)$ such that item $r$ was sold to bidder $a$ at or subsequent to the first time
that some item $r'$ was sold to some bidder $a'$ for some edge $(a',r')\in \sigma$. The edge $(a',r')$ is itself in $F^{\sigma}$.

\item Let $t^{\sigma}$ be the number of unsold instances just before the first time some edge along $\sigma$ was sold. That is,
$t^{\sigma}$ is equal to the number of instances matched in $F^{\sigma}$.

\item Let $\pi^{\sigma}$ be the price at which $Y^{\sigma}$ is computed.

\item Let $b^{\sigma}_a$ be the remaining budget for bidder $a$ before any item is sold in \textsc{Sell}($E^{\sigma}$) or \textsc{Sell}($a^{\sigma}$).

\end{itemize}
\end{Def}

We partition $\Sigma$ into two classes of trading paths:

\begin{itemize}
\item $\Sigma_E$ is the set of trading paths such that $\sigma \in \Sigma_E$ iff $Y^{\sigma}$ is some $E^{\sigma}$-avoid matching
used in \textsc{Sell}($E^{\sigma}$) (line~\ref{line:SellV} of Algorithm~\ref{alg:combclinch}).
\item $\Sigma_{\neg{E}}$ is the set of trading paths such that $\sigma \in \Sigma_{\neg{E}}$ iff
$Y^{\sigma}$ is some $a^{\sigma}$-avoid matching used in \textsc{Sell}($\{ a^{\sigma}\}$) (line~\ref{line:Selli} of Algorithm~\ref{alg:combclinch}).
\end{itemize}

\begin{lem}
\label{lem:SellV}
$\Sigma_E=\emptyset$.
\end{lem}

\begin{proof}
We need the following claim:

\begin{clm} Let $\sigma = (a_1,r_2, \ldots, a_{j-1},r_{j-1},a_j)\in \Sigma_E$ be a trading path, and let $(a_i,r_i)$
be the last edge belonging to $Y^{\sigma}$ along $\sigma$. Then the suffix of $\sigma$ starting at $a_i$, $(a_i, r_i, \ldots, a_j)$, is itself a
trading path.
\end{clm}
\begin{proof}
This follows as the valuation of $a_i$ is equal to the current price $\pi^\sigma$ when \textsc{Sell}($E^{\sigma}$) was executed, and the valuation of $a_1$
is greater than or equal to $\pi^{\sigma}$ as edge $(a_1,r_1)$ was unsold prior to this \textsc{Sell}($E^{\sigma}$) and does belong to the final $F^{\sigma}$.
\end{proof}

From the claim above we may assume, without loss of generality, that if $\Sigma_E \neq \emptyset$ then $\exists \sigma \in \Sigma_E$ such that the first edge along $\sigma$ was also the first edge sold amongst all edges of $\sigma$, furthermore, all subsequent edges do not belong to $Y^{\sigma}$.

As no further items will be sold to a bidder $a\in E^{\sigma}$ after this \textsc{Sell}($E^{\sigma}$), the number of items assigned to $E$-type bidders is equal for $Y^{\sigma}$ and $F^{\sigma}$. We seek a contradiction to the assumption that $Y^{\sigma}$ was an $E^{\sigma}$-avoid matching. Note that the matching $F^{\sigma}$ is an $E^{\sigma}$-avoid matching by itself because exactly the number of items assigned to $E$-type bidders in $Y^{\sigma}$ are being sold to them. We now show how to construct from $F^{\sigma}$ another matching that assigns less items to $E$-type bidders.

We show that the number of items assigned to bidder $a_1$ in $F^{\sigma}$ can be reduced by one by giving bidder $a_{k+1}$ item $r_k$ for $k=1, \ldots, j-1$. This is also a full matching but it remains to show that this does not exceed the capacity constraints $d_{a_j}$ of bidder $a_j$.

As $d_{a_j}=D_{a_j}$  for all $a\in A$ when \textsc{Sell}($E^{\sigma}$) is executed,  bidder $a_j$ has a remaining budget greater than or equal to $v_1$ at the conclusion of the auction, and each item assigned to bidder $a_j$ in $F^{\sigma}$ is sold to him at a price greater than or equal to $\pi^{\sigma} = v_1$, it follows that at the time of \textsc{Sell}($E^{\sigma}$) we have that $D_{a_j}$ is greater than the number of items assigned to $a_j$ in $F^{\sigma}$. Thus, we can increase the number of items allocated to $a_j$ by one without exceeding the demand constraint~$d_{a_j}$.

Now, note that $a_j$ is not an $E$-type bidder, and the new matching constructed assigns less items to $E$-type bidders than the matching $F^{\sigma}$. Hence, $F^{\sigma}$ is not an $E^{\sigma}$-avoid matching, and in turn neither $Y^{\sigma}$ is $E^{\sigma}$-avoid matching.
\end{proof}

We have shown that $\Sigma_E=\emptyset$. It remains to show that $\Sigma_{\neg{E}}=\emptyset$.

Assume $\Sigma_{\neg{E}} \neq \emptyset$. Order $\sigma \in \Sigma_{\neg{E}}$ by the first time at which some edge along $\sigma$ was sold. We know that
this occurs within some \textsc{Sell}($\{ a^{\sigma}\}$) for some $a^{\sigma}$ and that $a^{\sigma} \notin E$. Let us define $\sigma = (a_1, r_1, a_2, r_2, \ldots, a_{j-1},
r_{j-1}, a_j)$ be the last path in this order, and let $e=(a^{\sigma},r^{\sigma})=(a_i,r_i)$.

Recall that $Y^{\sigma}$ is the $a^{\sigma}$-avoid matching used when item $r^{\sigma}$ was sold to bidder $a^{\sigma}$. Also, $F^{\sigma}\subseteq H^*$ is the
set of edges added to $H^*$ in the course of the auction from this point on (including the current \textsc{Sell}($\{ a_i\}$)).

\begin{lem}
\label{lem:x}
Let $\sigma$, $a^{\sigma} = a_i$, and $r^{\sigma}=r_i$ be as above, then there was another full matching $X$ when $Y^{\sigma}$ was computed as an $a^{\sigma}$-avoid matching and $X$ has the following properties: 
\begin{enumerate}
\item[(a)] The suffix of $\sigma$ from $a_i$ to $a_j$:
\begin{equation*}
\sigma[a_i,\ldots,a_j] = (a_i,r_i,a_{i+1},r_{i+1},\ldots, a_{j-1},r_{j-1}, a_j),
\end{equation*}
is an alternating path
with respect to $X$ (i.e., edges $(a_k,r_k)$ where $i \leq k \leq j-1$ belong to~$X$).
\item[(b)] The number of items assigned to $a_i$ is equal in $X$ and in $Y^{\sigma}$.
\item[(c)] The number of items assigned to $a_j$ is equal in $X$ and in $F^{\sigma}$.
\end{enumerate}
\end{lem}

\begin{proof}
We use the notation $M(a)$ for the number of items assigned to bidder $a$ in a matching~$M$.
We know that $F^{\sigma}(a_i) \geq Y^{\sigma}(a_i)$ since there is otherwise a contradiction because $Y^{\sigma}$ is an $a_i$-avoid matching.

Notice that if $F^{\sigma}(a_i) = Y^{\sigma}(a_i)$, it is possible to choose $X = F^{\sigma}$ and the conditions above follow trivially.

Now, consider the case where $F^{\sigma}(a_i) > Y^{\sigma}(a_i)$. $Y^{\sigma}$ and $F^{\sigma}$ are both matchings that assign all~$t^{\sigma}$ instances, thus by Lemma \ref{lem:paths-and-cycles} we know that the symmetric difference between the two matchings can by expressed by sets of alternating paths. We consider the smallest such set, i.e., no two alternating paths can be concatenated. By Lemma~\ref{lem:simple-complex} we know that we can obtain a Pareto equivalent set of simple alternating paths with respect to~$F^\sigma$.
From the fact that $F^{\sigma}(a_i) > Y^{\sigma}(a_i)$, we can obtain $\delta = F^{\sigma}(a_i) - Y^{\sigma}(a_i)$ alternating paths that start from~$a_i$.
Consider one of this paths, $\tau = (a_i = g_1, s_1, g_2, s_2, \ldots, g_l)$, where $g_k$ are bidders, $s_k$ are items, $(g_k, s_k) \in F^{\sigma}$, and $(s_k, g_{k+1}) \in Y^{\sigma}$.

We argue that  $\sigma[a_i,\ldots,a_j]$ and $\tau$ are bidder disjoint, besides the first bidder $a_i$.
By contradiction, choose $u$ to be the first bidder other to $a_i$ in common between $\tau$ and $\sigma[a_i,\ldots,a_j]$. For some $i < k' \leq j$ and $1 < k \leq l$ we have $u = g_k = a_{k^{'}}$.
Let $\sigma^{'}$ be the concatenation of the prefix of $\sigma$ up to $a_i$, followed by the prefix of $\tau$ up to $g_k$ and followed by the suffix of $\sigma$ from $g_k = a_{k^{'}}$ to the end.
\begin{equation*}
\sigma^{'} = (a_1,r_1, \ldots, a_{i} = g_1, s_1, g_2, \ldots, g_k = a_{k^{'}}, r_{k^{'}}, a_{k^{'}+1}, \ldots, a_j)
\end{equation*}
This is a trading path in $F^{\sigma}$ and no edge is sold before $(a_i, r_i)$ in contradiction with the assumption that $\sigma$ is the last trading path in the defined order amongst all trading paths.
Thus, $\sigma[a_i,\ldots,a_j]$ and $\tau$ have the bidder~$a_i$ in common  and the other bidders along the paths are different. It could be possible that they have some items in common but this is no problem.\footnote{This is the case because if there are items in common but no bidders unless $a_i$, the edges that belong to $\sigma[a_i,\ldots,a_j]$ are not modified by any~$\tau$, so $\sigma[a_i,\ldots,a_j]$ will be an alternating path with respect to the $X$ we will define and the number of items assigned to bidder $a_j$ does not change for the same reason.}
For any such $\tau = (a_i = g_1, s_1, g_2, s_2, \ldots, g_l)$, we can move item $s_k$ from bidder $g_k$ to bidder $g_{k+1}$ where $1 \leq k \leq l-1$ without violating the demand of bidder $g_l$ because $s_{l-1}$ was assigned to $g_l$ in~$Y^{\sigma}$, and $g_l$ is not the first bidder in another alternating path.

As we can do so for all paths $\tau$, we obtain a new full matching $X$ by applying the swaps of all the alternating paths to $F^\sigma$.  $X$~assigns to $a_i$ the same number of items as $Y^{\sigma}$ and from the fact that $a_j$ does not appear in any $\tau$, the number of items assigned to him is again $F^{\sigma}(a_j)$.
\end{proof}

\begin{cor}
$\Sigma_{\neg{E}} = \emptyset$.
\end{cor}

\begin{proof}
Assume that $\Sigma_{\neg{E}} \neq \emptyset$, select $\sigma\in \Sigma_{\neg{E}}$ as in  Lemma~\ref{lem:x}, and let $a^{\sigma} = a_i$ and $r^{\sigma} = r_i$.
  We now seek to derive a contradiction as follows:

  \begin{itemize}
  \item when $Y^{\sigma}$ was computed there was also an alternative full matching $Y'$ with fewer
  items assigned to bidder $a_i$, contradicting the assumption that $Y^{\sigma}$ is an $a_i$-avoid matching, or
  
  \item we show that the remaining budget of bidder $a_j$ at the end of the auction, $b^*_{a_j}$, has $b^*_{a_j} < v_1$, contradicting the assumption that $\sigma$ is
  a trading path.
  \end{itemize}

Let $X$ be a matching as in Lemma~\ref{lem:x} and $F^{\sigma}$ be as defined in Definition~\ref{def:suppi}. Also, let $X(a)$, $F^{\sigma}(a)$, be the
number of items assigned to bidder $a$ in the full matchings $X$, $F^{\sigma}$, respectively.

We consider the following cases regarding $d_{a_j}$ when $Y^{\sigma}$, the $a_i$-avoid matching, was computed:

\begin{enumerate}

  \item[(a)] $d_{a_j}>X(a_j)$: Like in Lemma~\ref{lem:SellV}, we can decrease the number of items sold to $a_i$ by assigning item $r_k$ to bidder $a_{k+1}$ for $k=i, \ldots, j-1$, without exceeding the demand constraint~$d_{a_j}$.

  \item[(b)] $d_{a_j} = X(a_j)$: We show that $b^{\sigma}_{a_j} \leq (X(a_j)+1)\pi^{\sigma}$.

  \begin{itemize}

    \item $D_{a_j} = D^+_{a_j}$: Observe that $X(a_j)$ is smaller than the current number of unsold items~$\bar r$, and smaller than the cardinality of the interest set $S_{a_j}$ of bidder $a_j$. This follows because $r_{j-i}\in S_{a_j}$, but no instance of $r_{j-1}$ was sold to bidder~$a_j$ and an instance of $r_{j-1}$ is unsold at that time by the definition of~$\sigma$. Thus,
    \begin{equation*}
    X(a_j) = d_{a_j} =
    \left\lfloor \frac{b^{\sigma}_{a_j}}{\pi^{\sigma}} \right\rfloor > \frac{b^{\sigma}_{a_j}}{\pi^{\sigma}} -1,
    \end{equation*}
    and hence,
    \begin{equation*}
    b^{\sigma}_{a_j} < (X(a_j)+1)\pi^{\sigma}.
    \end{equation*}

    \item $D_{a_j} \neq D^+_{a_j}$: Observe that $a_j \notin E$ as $v_{a_j} > v_{a_i}$ and $a_i \notin E$.
    As $a_j\notin E$, the only reason that $D_{a_j} \neq D_{a_j}^+$ can be that the remaining budget of bidder $a_j$, $b^{\sigma}_{a_j}$, is
    an integer multiple of the current price $\pi^{\sigma}$. Then, $D_{a_j}^+ = D_{a_j}-1$ and by the same reason as above $D_{a_j} = \lfloor \frac{b^{\sigma}_{a_j}}{\pi^{\sigma}} \rfloor$. As $\lfloor \frac{b^{\sigma}_{a_j}}{\pi^{\sigma}} \rfloor =
    \frac{b^{\sigma}_{a_j}}{\pi^{\sigma}}$, it follows that
    \begin{equation*}
    X(a_j) = d_{a_j} \geq D_{a_j}^+ = D_{a_j} -1 = b^{\sigma}_{a_j}/\pi^{\sigma} -1,
    \end{equation*}
    and hence,
    \begin{equation*}
    b^{\sigma}_{a_j} \leq (X(a_j)+1)\pi^{\sigma}.
    \end{equation*}

  \end{itemize}

  Note that the current price $\pi^{\sigma}<v_{a_i}$ because we assume that $a_i$ was sold $r_i$ as a result of \textsc{Sell}($\{ a_i\}$) where $a_i$ is not an exiting bidder and not of \textsc{Sell}($E$). As $(a_i,r_i)$ was the first edge that was sold along $\sigma$, either $r_1$ was sold to $a_1$ for a price larger than $\pi^\sigma$, or $r_1$ was sold to $a_1$ at price $\pi^\sigma$ as a result of \textsc{Sell}($\{ a_1\}$) where $a_1$ is not an exiting bidder. Thus, $\pi^\sigma < v_{a_1}$.
  
By Condition~(c) of Lemma~\ref{lem:x} we can deduce that
\begin{equation*}
b^{\sigma}_{a_j} \leq (X(a_j) + 1) \pi^{\sigma} = (F^{\sigma}(a_j)+1)\pi^{\sigma}.
\end{equation*}
Bidder $a_j$ are sold exactly $F^{\sigma}(a_j)$ items at a price not lower than $\pi^{\sigma}$. Hence, at the end of the auction the remaining budget $b^*_{a_j}$ of
bidder $a_j$ is lesser than or equal to $\pi^{\sigma}$.  This contradicts the assumption that $\sigma$ is a trading path since
\begin{equation*}
b^*_{a_j} \leq \pi^{\sigma} <  v_{a_1}.\qedhere
\end{equation*}
\end{enumerate}
\end{proof}

\section{Proof of Theorem~7}

We show the impossibility result for the case of two bidders and two items. In order to demonstrate this theorem we use the following two lemmas. The first one follows closely the one of \cite{FLSS11}.

\begin{lem}\label{lem:imphelp1}

Consider any incentive compatible, Pareto optimal, bidder rational, and auctioneer rational multi-unit auction with marginal valuation functions that produces an allocation $(H, p)$. If bidder $1$ has positive valuations $v_1(1)$ and $v_1(2)$ and wins both items then the payment~$p_2$ by bidder~$2$ is zero.

\end{lem}

\begin{proof}
First, consider the case when $v_2(1)=0,v_2(2) = 0$. Then any Pareto optimal auction has to assign both items to bidder $1$. If any of the items were to be left 
unassigned,
or would be assigned to bidder $2$, we could assign it to bidder $1$, without changing any payment. This does not change the utility of bidder $2$, nor the utility of the auctioneer, but would 
strictly increase the utility of bidder $1$.
Because of incentive compatibility, bidder $1$ pays $p_1 = 0$. Otherwise, bidder $1$ could reduce his reported valuation and attain the item at a lower price. If follows from bidder
rationality that $p_2 \leq 0$. However, it follows from auctioneer rationality that bidder~1 must pay zero, as $-p_2 \leq p_1 = 0$.
Now, consider the case when both bidders have nonzero valuations. Then for every instance in which bidder $2$ gets no items it must be that $p_2 = 0$. By incentive compatibility his payment cannot depend on his
valuation, and when bidder $2$ reported a valuation of zero then $p_2$ was zero.
\end{proof}

\begin{lem}\label{lem:imphelp2}
Consider any incentive compatible, Pareto optimal, bidder rational and auctioneer rational multi-unit auction with budgets and marginal valuation functions: if $v_1(|R|) > v_2(1)$ and $\sum_{j=1}^{|R|} v_2(j) \leq b_1$, where $|R|$ is the number of items, then all items will be assigned to bidder $1$.
\end{lem}

\begin{proof}
The lemma follows directly from the definition of Pareto optimality. If an item is assigned to bidder $2$, it follows from the assumptions that:
\begin{itemize}

\item $v_1(|R|) > v_2(1)$

\item $\sum_{j=1}^{|R|} v_2(j) \leq b_1 \Rightarrow \sum_{j=1}^{|R|} v_2(j) - \sum_{j=2}^{|R|} v_2(j) \leq b_1 - \sum_{j=2}^{|R|} v_2(j) \Rightarrow b^*_1 \geq v_2(1)$

\end{itemize}

The last implication follows from incentive compatibility and bidder rationality. This implies the existence of a trading path from bidder $2$ to bidder $1$ and from Theorem \ref{thm:paropt} we know that the allocation is not Pareto optimal. It follows that no item can be assigned to bidder~2.  
\end{proof}

Moreover, we will use the following theorem from \cite{DLNcorrected}. Please note that their definition of {\em individual-rationality} corresponds to our definition of {\em bidder rationality} and their definition of {\em no-positive-transfers} corresponds to our definition of {\em auctioneer rationality}.

\begin{thm}[Theorem 4.1 in \cite{DLNcorrected}]\label{thm:uniqueness}
Let $A$ be a deterministic truthful mechanism for $m$ items and $2$ players with known budgets $b_1$ and $b_2$ that are generic. Assume that $A$ satisfies Pareto-optimality, individual-rationality, and no-positive-transfers. Then if $v_1 \neq v_2$ the outcome of $A$ coincides with that of the clinching auction.
\end{thm}

We are now ready to prove the main theorem:

\begin{proof}[\bf Proof of Theorem \ref{thm:impossibility2}]
Consider the case of  two bidders and two items such that: $v_1(1) = 5$, $v_1(2) = 5$, $b_1 = 3$, $v_2(1) = 2$, $v_2(2) = 2$, and $b_2 = 11$.
From the Theorem \ref{thm:uniqueness} we know that the outcome has to be equal to the outcome of the clinching auction of \cite{DBLP:conf/focs/DobzinskiLN08}. Thus, both  bidders receive one item at prices $p_1 = 2$ and $p_2 = 1.5$ and the utility for the two bidders are $u_1 = 5 - 2 = 3$ and $u_2 = 2 - 1.5 = 0.5$.

Now, assume that the true marginal valuations for bidder $2$ are $v_2(1) = 2$ and $v_2(2) = 1$. It follows from Lemma \ref{lem:imphelp2} that all items are assigned to bidder $1$.  Thus, from Lemma \ref{lem:imphelp1}, the payment of bidder $2$ will be zero and his utility will be zero too.  We conclude that bidder $2$ has an incentive to lie in his marginal valuation function.
\end{proof}

\end{document}